\documentclass[journal]{IEEEtran}

\IEEEoverridecommandlockouts                              

\title{\LARGE \bf
Resilient Scheduling of Control Software Updates \\ in Radial Power Distribution Systems
}

\author{Kin Cheong Sou and Henrik Sandberg  
\thanks{K.C. Sou is with the Department of Electrical Engineering, National Sun Yat-sen University, Taiwan
        {\tt\small sou12@mail.nsysu.edu.tw}. K.C.~Sou was funded by the Ministry of Science and Technology (MOST) of Taiwan: Optimal dispatch of flexibility resources in power systems with high penetration of renewables, MOST 109-2221-E-110-016-MY2.}%
\thanks{H. Sandberg is with the Division of Decision and Control Systems, EECS, KTH Royal Institute of Technology, Stockholm, Sweden
        {\tt\small hsan@kth.se}. H.~Sandberg was funded by the Swedish Energy Agency and ERA-Net Smart Energy Systems (project RESili8, grant agreement No~883973).}%
}

\usepackage{graphicx,amsmath,xcolor,amssymb,amsthm,multirow,cite}
\usepackage{algpseudocode,algorithm}

\newcommand*{\R}{\mathbb{R}}
\newtheorem{proposition}{Proposition}
\newtheorem{lemma}{Lemma}

\algrenewcommand\algorithmicrequire{\textbf{Input:}}
\algrenewcommand\algorithmicensure{\textbf{Output:}}

\newcommand{\minuseq}{\mathrel{-}=}

\begin{document}

\maketitle
\thispagestyle{empty}
\pagestyle{empty}

\begin{abstract}

In response to newly found security vulnerabilities, or as part of a moving target defense, a fast and safe control software update scheme for networked control systems is highly desirable. We here develop such a scheme for intelligent electronic devices (IEDs) in power distribution systems, which is a solution to the so-called software update rollout problem. This problem seeks to minimize the makespan of the software rollout, while guaranteeing safety in voltage and current at all buses and lines despite possible worst-case update failure where malfunctioning IEDs may inject harmful amounts of power into the system. Based on the nonlinear DistFlow equations, we derive linear relations relating software update decisions to the worst-case voltages and currents, leading to a decision model both tractable and more accurate than previous models based on the popular linearized DistFlow equations. Under reasonable protection assumptions, the rollout problem can be formulated as a vector bin packing problem and instances can be built and solved using scalable computations. Using realistic benchmarks including one with 10,476 buses, we demonstrate that the proposed method can generate safe and effective rollout schedules in real-time.


\end{abstract}

\section{INTRODUCTION}

Critical infrastructures such as power and water distribution networks, and cyber-physical systems (CPSs) in general, have in the past decade become targets of cyberattacks~\cite{hemsley+18,cardenas2019}. Simultaneously, cybersecurity features in these systems are lacking. Control systems have been designed to meet safety requirements, often meaning a preference for low complexity solutions that minimize time delay and prevents the adoption of many standard computer security features. It has now become urgent to develop security solutions that take the special safety requirements of CPSs into account.

An essential difference between CPS and regular computer security is that ``software patching and frequent updates, are not well suited for control systems''~\cite{cardenas+08}. Software updates are problematic in that they may require a complex reboot, which may come with an operational cost and a safety risk. It is not uncommon that control software is updated infrequently, and systems may run for long time even with known vulnerabilities. A practically relevant research problem is to develop control algorithms that can be safely updated in real time, to patch newly discovered security weaknesses or as part of moving target or software rejuvenation defense~\cite{griffioen2019moving,kanellopoulos+20,Romagnoli+23}.

This paper considers a control software update rollout problem for power distribution systems introduced in \cite{de2020minimum}, and extends the results in \cite{sou2022resilient}. The rollout problem seeks to arrange the software updates of IEDs (e.g., smart inverters) into a minimum-time schedule, while guaranteeing power system operational safety despite worst-case update failure. The problem is not solvable using traditional maintenance planning techniques~\cite{froger+16} due to its intricate cyber-physical relationships involved. A central challenge here is an accurate yet tractable relation between the software update decisions and their worst-case consequences to the system states (e.g., voltages and currents). We derive the desired relation using the \emph{nonlinear} DistFlow equations~\cite{BaranWu89,BaranWu89B} as opposed to the popular linearized DistFlow equations adopted in \cite{de2020minimum}. The nonlinear DistFlow equations more accurately describe voltages and explicitly model line currents absent in the linearized equations. We numerically demonstrate the ramifications of the DistFlow models in maintaining safety standard. Compared with \cite{sou2022resilient}, this paper presents major upgrades to the current and voltage safety limit constraints to reduce the conservatism of the rollout problem. Our numerical case studies indicate that the rollout schedules obtained in this paper are effective over a wide variety of test cases. In addition, this paper features streamlined computation procedures to set up and solve the rollout problem in real-time. These include fixed point iterations to compute voltage and current bounds (required to build a rollout problem instance) and a modified heuristic algorithm to solve the associated bin-packing problem. These innovations eliminate the need for mixed integer (bi)linear programming in \cite{sou2022resilient}, enabling practical real-time large-scale rollout scheduling. For instance, our procedure can find a non-trivial schedule in a 10,476-bus case in less than 3 seconds. This was impossible in \cite{sou2022resilient}. Our novel voltage and current bounds are potentially useful in robust and contingency-based AC optimal power flow problems, which has received significant interest in the past few years. See, for example, \cite{Carleton+17,Molzahn+18,Louca+19,wu2018robust,coffrin2018relaxations,liu2021optimal}.

Smart distribution grids play an essential role in society's transition into a net-zero energy system and thereby achieve highly set climate goals. The transition requires integration of a wide variety of controllable energy devices, both for consumption and for supply. The interface between power supply (grid) and demand (consumer) is sometimes referred to as the \emph{grid edge}~\cite{siemens19}. The grid edge will turn distribution power grids, traditionally not very automated, into complex distributed control systems. The many local control loops in the grid edge simultaneously present an increased attack surface, and will require resilient and systematic controller update/patching schemes to handle newly discovered vulnerabilities or to respond to contingencies~\cite{kintzler+18,shelar2016security,shelar2021evaluating}, while accounting for possible adverse effects of failed control updates.


This paper is organized as follows. Section~\ref{sec:system_modeling} introduces notations, distribution system model and its operational requirements. In Section~\ref{sec:problem_statement} the software update rollout problem is defined. The exact mathematical model is unfit for real-time applications. Thus, a tractable approximation is derived. Section~\ref{sec:solution} summarizes the solution procedure for the rollout problem. Section~\ref{sec:demonstration} presents case studies to demonstrate that the proposed procedure is time-efficient and is able to deliver effective and safe rollout schedules for large systems.

\section{System Modeling} \label{sec:system_modeling}

\subsection{Mathematical Notations} \label{subsec:notations}
We define the following notations: the vector ${\bf 1}$ (resp., ${\bf 0}$) is the all-one (resp., all-zero) vector, and $e_i$ for $i = 1,2,\ldots$ is the $i^{\text{th}}$ unit vector. By default, the symbol $j$ means $\sqrt{-1}$. For a set of positive integers $\mathcal{N}$ (details in Section~\ref{subsec:distribution_system_notations}) and $\mathcal{I} \subseteq \mathcal{N}$, the symbol ${\bf 1}_{\mathcal{I}} \in \{0,1\}^{\mathcal{N}}$ is the 0-1 binary indicator vector with support $\mathcal{I}$. For any vector $v$, $\text{diag}(v)$ or $D_v$ are the diagonal matrix with the diagonal entries defined by $v$. For any vector $v$ and scalar $a$ (positive or negative), the symbol $v^a$ denotes the entry-wise exponentiation (i.e., the $i$-th entry of $v^a$ is $v_i^a$). For any two vectors $x$ and $y$ of the same length, the symbol $x \odot y$ denotes the entry-wise (Hadamard) product of $x$ and $y$. For any two matrices $A$ and $B$, the symbol $A \otimes B$ denotes the Kronecker product of $A$ and $B$.

\subsection{Power Distribution System Modeling} \label{subsec:distribution_system_notations}

We consider single-phase radial power distribution systems. This can potentially be interpreted as the positive sequence approximation of a three-phase system. We assume that the system operates in steady-state and hence all electrical quantities (e.g., current, voltage) can be represented using per unit phasors. We use the following notations to describe the system:
\begin{itemize}
    \item $N$: number of non-reference buses (also number of lines)
    
    \item $\mathcal{N} = \{1,\ldots,N\}$: the set of all non-reference buses. Bus 0 is the reference bus (also called the slack bus).
    
    \item $\mathcal{L} = \{1, \ldots,N\}$: the set of all lines. Each line has a reference direction for line current and line power flow. Reference direction points away from bus 0. A line is labeled by the bus that it points to (e.g., line $n$ points to bus $n$ and no other line points to $n$ due to radiality).
    
    \item $\pi_n \in \mathcal{N} \cup \{0\}$: the ``parent'' bus of $n \in \mathcal{N}$ (so that line $n$ goes from $\pi_n$ to $n$).

    \item $d(n) \subseteq \mathcal{N}$: for any $n \in \mathcal{N}$, $d(n)$ denotes the set of descendants of $n$ including $n$ itself. That is, $m \in d(n)$ if and only if the (only) path from bus 0 to $m$ traverses $n$.
    
    \item $\tilde{A}$: $N \times (N+1)$ line-bus incidence matrix. For $(m,n) \in \mathcal{L} \times (\mathcal{N} \cup \{0\})$, $\tilde{A}_{mn} = 1$ if line $m$ leaves bus $n$, $\tilde{A}_{mn} = -1$ if line $m$ enters bus $n$ and $\tilde{A}_{mn} = 0$ otherwise.
    
    \item $A$: $N \times N$ submatrix of $\tilde{A}$ with the first column $a_0$ removed (i.e., $\tilde{A} = \begin{bmatrix} a_0 & A \end{bmatrix}$).

    \item $\mathcal{D}$: $\{0,1\}^{N \times N}$ descendant matrix so that $\mathcal{D}_{nm} = 1$ if and only if $m \in d(n)$. It can be verified that $\mathcal{D}^\top = -A^{-1}$ when row $k$ of $A$ corresponds to line $k$ for all $k \in \mathcal{L}$.
    
    \item All shunt elements are ignored and $z = r + jx$ is an $N$-vector of line series impedance; $r > {\bf 0}$ is vector of series resistances and $x > {\bf 0}$ is vector of series reactances.
    
    \item $D_r$, $D_x$: $N \times N$ positive diagonal matrices with $D_r = \text{diag}(r)$ and $D_x = \text{diag}(x)$. Also, $D_z := D_r + j D_x$.
    
    \item $P$, $Q$, $S$: $N$-vectors of active, reactive and apparent power flows for all lines. $P$, $Q$ and $S$ are defined at the sending ends of the lines (e.g., $P_n$ is defined at bus $\pi_n$).
    
    \item $p$, $q$, $s$: $N$-vectors of net injections of active power, reactive power, apparent power for non-reference buses.
    
    \item $p = p^G - p^L$, $q = q^G - q^L$ with $p^G$, $q^G$, $p^L$, $q^L$ vectors of active power generation, reactive power generation, active power load, reactive power load for non-reference buses.
    
    \item $\ell$: $N$-vector of squared magnitude line currents.
    
    \item $v$: $N$-vector of magnitude voltage of non-reference buses.
    
    \item $\nu$: $N$-vector of squared magnitude voltage (i.e., $\nu = v^2$). At bus 0, squared magnitude voltage $\nu_0$ is constant.
\end{itemize}
For all vectors above, subscript $n$ means the $n^{\text{th}}$ entry of the vector. For instance, $p^G_n$ is the active power generation injection to bus $n \in \mathcal{N}$, while $\ell_n$ is the squared current for line $n \in \mathcal{L}$.
Fig.~\ref{fig:line_segment} illustrates a segment of the distribution system with the relevant electrical quantities.
\begin{figure}
    \centering
    \includegraphics[width=60mm]{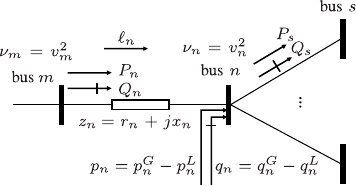}
    \caption{A segment of distribution system with the relevant quantities}
    \label{fig:line_segment}
\end{figure}

For a radial distribution system with shunt elements ignored, the electrical quantities $\nu, \ell, P, Q, p, q$ are related by the (nonlinear) DistFlow equation \cite{BaranWu89, BaranWu89B} as follows: 
\begin{subequations}
\label{eqn:DF}
\begin{align}
    p &= A^{\top} P + D_r \ell \label{eqn:DF_Pline} \\
    q &= A^\top Q + D_x \ell \label{eqn:DF_Qline} \\
    A\nu + \nu_0 a_0 &= 2 D_r P + 2 D_x Q - (D_r^2 + D_x^2) \ell \label{eqn:DF_Ohm} \\
    \ell_n &= \frac{P_n^2+Q_n^2}{\nu_{\pi_n}}, \quad \forall n = 1, 2, \ldots, N \label{eqn:DF_S}
\end{align}
\end{subequations}
Note that $A^{-1} = -\mathcal{D}^\top$ and $A^{-1} a_0 = -{\bf 1}$ (since $\tilde{A} {\bf 1} = {\bf 0}$). The variables $P$ and $Q$ can be eliminated from \eqref{eqn:DF} to obtain
\begin{subequations} \label{eqn:DF_vi}
\begin{align}
    \nu &= \nu_0 {\bf 1} + 2 R p + 2 X q + M \ell \label{eqn:DF_nu} \\
    \ell &= (\nu)^{-1} \odot\big| \mathcal{D}(p + j q) - (\mathcal{D} - I) D_z \ell \; \big|^2 \label{eqn:DF_l}
\end{align}
\end{subequations}
where
\begin{align} \label{eqn:RXM}
\begin{split}
    R &= \mathcal{D}^\top D_r \mathcal{D} \\
    X &= \mathcal{D}^\top D_x \mathcal{D} \\
    M &= \mathcal{D}^\top D_r (I - 2 \mathcal{D}) D_r + \mathcal{D}^\top D_x (I - 2 \mathcal{D}) D_x
\end{split}
\end{align}
Equation \eqref{eqn:DF_vi} implies that $\nu$ and $\ell$ are functions of net power injections $(p,q)$, motivating the shorthand $\nu = \nu(p,q)$ and $\ell = \ell(p,q)$. Also, since $\mathcal{D} \in \{0,1\}^{N \times N}$ with unit diagonal entries and $r$, $x$ are positive, \eqref{eqn:RXM} implies that all entries of $R$ and $X$ are nonnegative, while all entries of $M$ are non-positive.

\subsection{Distribution System Setting and Grid Code} \label{subsec:DSS_grid_code}
We assume that every non-reference bus in the distribution system is equipped with a smart inverter to control and monitor its distributed generation. The inverter-interfaced distributed generation is $p^G_n \ge 0$ for bus $n \in \mathcal{N}$. In addition, the inverter can provide reactive power support with reactive power injection $q^G_n$, which is not sign restricted. We assume the rating of the inverter at bus $n$ is $C_n$ restricting its apparent power injection. If bus $n$ has no generation, then $C_n = 0$. The vector of all inverter ratings is denoted by $C$. Thus, the active and reactive power injections of the inverters satisfy
\begin{equation} \label{eqn:inverter_constraints}
    p^G_n \ge 0, \quad \big| p^G_n + j q^G_n \big| \le C_n, \quad \forall n \in \mathcal{N}
\end{equation}
It is assumed that the operator knows the inverter generation set points denoted by $\hat{p}_n^G$ and $\hat{q}_n^G$ and the load estimates denoted by $\hat{p}^L_n$ and $\hat{q}^L_n$. 
In normal operating conditions, inverter set points $\hat{p}^G$ and $\hat{q}^G$ are chosen so that the grid code is satisfied. In this paper, the grid code specifies that the steady-state $\nu$ and $\ell$ should satisfy the following safety limits:
\begin{subequations}
\label{eqn:lv_bounds}
\begin{align}
    & \underline{\nu} \le \nu \le \overline{\nu}, \quad \text{with given} \;\; \underline{\nu} \in \R^N \; \text{and given} \; \overline{\nu} \in \R^N, \label{eqn:v_bounds}\\
    & 0 \le \ell \le \overline{\ell}, \quad \text{with given} \;\; \overline{\ell} \in \R^N \label{eqn:l_bounds}
\end{align}
\end{subequations}
We note that $\ell \ge 0$ is not explicitly enforced as it is implied by \eqref{eqn:v_bounds} and \eqref{eqn:DF_S}. Typical values of the safety limits are: 0.9 pu for $(\underline{\nu})^{1/2}$, 1.1 pu for $(\overline{\nu})^{1/2}$ and a few hundred amperes (e.g., 600 A) for $(\overline{\ell})^{1/2}$.

\section{Problem Statement and Formulation} \label{sec:problem_statement}

\subsection{Rollout Problem Setting and Statement}

As explained in the introduction, we envision a scenario where, due to security or operational considerations, the operator needs to remotely update the software or firmware of all inverters through some communication network.
To minimize possible disruption, the operator wishes to schedule the updates to finish as soon as possible (e.g., all at once if possible). However, some or all updates may fail. 
Software update failure for the inverter at a bus is modeled by the condition that, instead of following its power set point command, the inverter may inject an uncontrollable amount of power subjected to its physical limit in~\eqref{eqn:inverter_constraints}. We assume that each bus is equipped with a relay to promptly detect out-of-range voltage or current (usually within a few milliseconds). Then the breaker or other protection will act to isolate the faulty inverter from the rest of the system and the update for the inverter is rolled back. We assume that $\Delta \tau$ is the universal fault clearing time for the system. In other words, if an inverter experiences software update failure, it may inject uncontrolled amount of power for up to $\Delta \tau$ seconds, and then its power output is back to normal.

It is not desirable to schedule too many updates at the same time lest the updates may fail simultaneously and inject a dangerous level of uncontrolled power causing power outage or even risk to human lives. Thus, the operator is faced with an update scheduling problem where s/he seeks to minimize the makespan of the entire process while guaranteeing the safe operation of the system despite the worst-case update failure due to the schedule. This problem is referred to as the \emph{software update rollout problem} or the \emph{rollout problem} for short.

Safe operation means that the steady-state squared voltage $\nu$ and steady-state squared line current $\ell$ satisfy the grid code constraints in~\eqref{eqn:lv_bounds}. The transient voltage and current from a normal operating condition to the situation after update failure may be different from their respective steady-state values. However, we assume that the system possesses sufficient ride-through capability so that if~\eqref{eqn:lv_bounds} is satisfied in steady states then the transient voltage and current will satisfy their respective transient safety requirements.

Formally, the safety requirement of the rollout problem can be stated as follows:
For any $\mathcal{I} \subseteq \mathcal{N}$, define $\mathcal{S}(\mathcal{I})$ by
\begin{equation} \label{eqn:SI}
    \begin{split}
        \mathcal{S}(\mathcal{I}) := &\big\{(p,q) \mid p = p^G - \hat{p}^L, \; q = q^G - \hat{q}^L \\
        & \big| p^G_n + j q^G_n \big| \le C_n, \; p^G_n \ge 0, \; \text{if $n \in \mathcal{I}$} \\
        & p^G_n = \hat{p}^G_n, \; q^G_n = \hat{q}^G_n, \; \text{if $n \in \mathcal{N} \setminus \mathcal{I}$} \big\}
    \end{split}
\end{equation}
as the set of all possible $(p,q)$ due to update failures at buses in $\mathcal{I} \subseteq \mathcal{N}$. Then, it is required that all solutions $\nu(p,q)$ and $\ell(p,q)$ of \eqref{eqn:DF} due to $(p,q) \in \mathcal{S}(\mathcal{I})$ satisfy the safety limits in \eqref{eqn:lv_bounds}. Subsequently, we derive linear relationships between ${\bf 1}_{\mathcal{I}}$ (i.e., the decision indicator vector of the rollout problem) and the worst-case $\nu(p,q)$ and $\ell(p,q)$ for $(p,q) \in \mathcal{S}(\mathcal{I})$ (used to describe the safety constraints in the rollout problem).


\subsection{Universal voltage lower bound and current upper bound} \label{subsec:vlb_iub}
To derive the safety constraints of the rollout problem, we need bounds of $\nu$ and $\ell$ as simple functions of $(p,q)$. Since $\ell \ge 0$ and $M$ has non-positive entries in \eqref{eqn:DF_vi}, an immediate upper bound of $\nu$ arises when the term $M \ell$ in \eqref{eqn:DF_nu} is ignored. This will be discussed in Section~\ref{sec:v_UB}. On the other hand, lower bound of $\nu$ and upper bound of $\ell$ are intertwined. To obtain the bounds, for any $n \in \mathcal{N}$ we denote squared voltage lower bound $\nu^L_n$ and squared current upper bound $\ell^U_n$ as
\begin{subequations} \label{opt:vi_bounds}
    \begin{align}
        \nu^L_n &:= \min\limits_{p, q, P, Q, \nu, \ell} \; \nu_n \quad \text{subject to \eqref{eqn:DF} and \eqref{eqn:inverter_constraints}} \label{opt:vi_bounds_V} \\
        \ell^U_n &:= \max\limits_{p, q, P, Q, \nu, \ell} \; \ell_n \quad \text{subject to \eqref{eqn:DF} and \eqref{eqn:inverter_constraints}} \label{opt:vi_bounds_l}
    \end{align}
\end{subequations}
Then, Sections~\ref{sec:l_UB} and \ref{sec:v_LB} discuss how $\nu^L$ and $\ell^U$ can be used in \eqref{eqn:DF_vi} to obtain the desired voltage and current bounds with respect to $(p,q)$. Problem \eqref{opt:vi_bounds} can be solved as a non-convex bilinear program (due to \eqref{eqn:DF_S}). However, this may not be suitable for real-time large-scale applications. Instead, we propose to quickly estimate lower bounds of $\nu^L$ and upper bounds of $\ell^U$ using the fixed point iterations in Algorithm~\ref{alg:fixed_point}.
\begin{algorithm}
\caption{Fixed point iterations to bound $\nu^L$ and $\ell^U$}
\label{alg:fixed_point}
\begin{algorithmic}[1]
    \Require Matrices $\mathcal{D}$ and $D_z$ in Section~\ref{subsec:distribution_system_notations}, vectors $\hat{p}^L$, $\hat{q}^L$, $C$ in Section~\ref{subsec:DSS_grid_code}, $R$, $X$, $M$ in \eqref{eqn:RXM} and tolerance $\epsilon > 0$
    \Ensure $\hat{\nu}^L \le \nu^L$ and $\hat{\ell}^U \ge \ell^U$
    \State Denote $a := \mathcal{D} \hat{p}^L$, $b := \mathcal{D} \hat{q}^L$, $r := \mathcal{D} C$, $\theta \in \R^N$ s.t.
    \begin{align*}
        \theta_n := \begin{cases}
            \tan^{-1}(b_n/a_n), & \text{if $a_n \neq 0$ or $b_n \neq 0$} \\
            0, & \text{if $a_n = b_n = 0$}
        \end{cases}
    \end{align*}
    \State Define
    \begin{align*}
        \bar{S} := \max \Big\{ &|a \! + \! j(b \!+ \! r)|, |a \! + \! j(b \!- \!r)|, |a \!+\!jb \!- \! r \! \odot \! e^{j \theta}| \Big\}
    \end{align*}
    where ``max'' is interpreted row-wise
    \State Initialize $v^0 \in \R^N$ and $i^0 \in \R^N$ such that $v^0 \le (\nu^L)^{1/2}$ and $i^0 \ge (i^U)^{1/2}$ and $k \leftarrow 0$
    \While{$\|v^k - v^{k-1}\| > \epsilon \;\; \text{or} \;\; \|i^k - i^{k-1}\| > \epsilon$}
    \State $v^{k+1} = \big(\nu_0 {\bf 1} \! - \! 2 R \hat{p}^L - 2 X (\hat{q}^L \! + \! C) \! + \! M (i^k)^2 \big)^{1/2}$
    \State $i^{k+1} = \big(v^k)^{-1} \odot (\bar{S} + |(\mathcal{D} - I) D_z (i^k)^2 |\big)$
    \State $k \leftarrow k+1$
    \EndWhile
    \State Return $\hat{\nu}^L = (v^k)^2$ and $\hat{\ell}^U = (i^k)^2$
\end{algorithmic}
\end{algorithm}

Three questions pertain to Algorithm~\ref{alg:fixed_point}: (a) Does Algorithm~\ref{alg:fixed_point} converge? (b) Is it true that $\hat{\nu}^L \approx \nu^L$ and $\hat{\ell}^U \approx \ell^U$? (c) Is it true that $\hat{\nu}^L \le \nu^L$ and $\hat{\ell}^U \ge \ell^U$? Question (a) regarding convergence can be verified online. Further, our extensive simulations indicate quick convergence as long as the initial guess $i^0$ is not too large (so $v^k$ stays real-valued). As explained in Section~\ref{subsec:bound_quality}, question (b) regarding approximation quality can also be checked online. However, question (c) regarding bounding property cannot be verified numerically and it must be proven analytically. Next, we provide an intuition to justify (b) and (c). Then, a statement is presented to prove (c).

Suppose the while-loop of line 4-8 in Algorithm~\ref{alg:fixed_point} converge with limits $\ell^\star := (\lim\limits_{k \to \infty} i^k)^2$ and $\nu^\star := (\lim\limits_{k \to \infty} v^k)^2$. Then, $i^\star$ and $\nu^\star$ are characterized by
\begin{subequations} \label{eqn:FP_vi}
\begin{align}
    \nu^\star &= \nu_0 {\bf 1} - 2 R \hat{p}^L - 2 X (\hat{q}^L + C) + M \ell^\star \label{eqn:FP_v} \\
    \ell^\star &= {(\nu^\star)}^{-1} \odot \big(\bar{S} + |(\mathcal{D} - I) D_z \ell^\star |\big)^2 \label{eqn:FP_i}
\end{align}
\end{subequations}
which resembles \eqref{eqn:DF_vi}, except that the right-hand-sides of \eqref{eqn:FP_vi} are optimized with respect to $(p,q)$ subject to \eqref{eqn:inverter_constraints}. Indeed, due to the fact that $R$ and $X$ in \eqref{eqn:RXM} have nonnegative entries and Lemma~\ref{thm:abr_max} in Appendix A the following holds: for any $n \in \mathcal{N}$, entries $\nu^\star_n$ in \eqref{eqn:FP_v} and $\ell^\star_n$ in \eqref{eqn:FP_i} can be rewritten as
\begin{align} \label{eqn:FP_vi1}
\begin{split}
    \nu^\star_n &= \underset{(p,q) \; \text{in} \; \text{\eqref{eqn:inverter_constraints}}}{\min} e_n^\top \left( \nu_0 {\bf 1} + 2 R p + 2 X q + M \ell^\star \right) \\
    \ell^\star_n &= \underset{(p,q) \; \text{in} \; \text{\eqref{eqn:inverter_constraints}}}{\max}  e_n^\top \big( {(\nu^\star)}^{-1} \odot \big(|\mathcal{D}(p + j q)| \! + \! |(\mathcal{D} - I) D_z \ell^\star |\big)^2 \big) \\
    &\ge \underset{(p,q) \; \text{in} \; \text{\eqref{eqn:inverter_constraints}}}{\max} e_n^\top \big( {(\nu^\star)}^{-1} \odot \big|\mathcal{D}(p + j q) - (\mathcal{D} - I) D_z \ell^\star \big|^2 \big)
\end{split}
\end{align}
where ``$(p,q)$ in \eqref{eqn:inverter_constraints}'' means $(p,q)$ satisfying \eqref{eqn:inverter_constraints} in \eqref{eqn:FP_vi1}. The relations in \eqref{eqn:FP_vi1} have two implications. Firstly, the resemblance of \eqref{eqn:FP_vi1} to \eqref{eqn:DF_vi} (equivalent to \eqref{eqn:DF}) and the optimization suggest that $\nu^\star \approx \nu^L$ and $\ell^\star \approx \ell^U$ (hence $\hat{\nu}^L \approx \nu^L$ and $\hat{\ell}^U \approx \ell^U$). Secondly, the optimization in \eqref{eqn:FP_vi1} suggests the inequalities $\nu^\star \le \nu^L$ and $\ell^\star \ge \ell^U$ and hence $\hat{\nu}^L \le \nu^L$ and $\hat{\ell}^U \ge \ell^U$. Indeed, this can be verified by the following statement, whose proof can be found in Appendix B.
\begin{proposition} \label{thm:fixed_point}
    For Algorithm~\ref{alg:fixed_point}, it holds that $v^k \le (\nu^L)^{1/2}$ and $i^k \ge (\ell^U)^{1/2}$ for all $k \ge 0$.
\end{proposition}
In particular, Proposition~\ref{thm:fixed_point} specifies that $(\hat{\nu}^L, \hat{\ell}^U)$ as returned by Algorithm~\ref{alg:fixed_point} satisfy the desired bounding inequalities. Proposition~\ref{thm:fixed_point} also guarantees that, in case Algorithm~\ref{alg:fixed_point} does not converge (though we never encounter this in our extensive simulations), any premature solution from Algorithm~\ref{alg:fixed_point} still satisfies the bounding inequalities. In practice, Algorithm~\ref{alg:fixed_point} converges quickly even if the initialization conditions $v^0 \le (\nu^L)^{1/2}$ and $i^0 \ge (i^U)^{1/2}$ are not followed. As is standard in load flow analysis, we adopt a ``flat start'' of $v^0 = {\bf 1}$ and $i^0 = {\bf 0}$ in here. For convenience, the shorthand $\hat{v}^L := (\hat{\nu}^L)^{1/2}$ and $\hat{i}^U := (\hat{\ell}^U)^{1/2}$ will be used subsequently.

\subsection{Expressions Related to Line Current Upper Limit} \label{sec:l_UB}
This section develops two separate results -- (a) an upper bound of line current affinely dependent on ${\bf 1}_{\mathcal{I}}$ and (b) given any line current upper limit $\bar{\ell} \in \R^N$ a set of linear inequalities on ${\bf 1}_{\mathcal{I}}$ to ensure that $\ell(p,q) \le \bar{\ell}$ for all $(p,q) \in \mathcal{S}(\mathcal{I})$.

{\bf Affine current upper bound}: Let $\ell = \ell(p,q)$ from \eqref{eqn:DF_vi}. Then, $\ell^{1/2}$ is the vector of magnitude line currents and hence by Kirchhoff's current law
\begin{align} \label{eqn:In_UB}
    \sqrt{\ell_n} = \left| \sum\limits_{m \in d(n)} i'_m \right| \le \!\!\! \sum\limits_{m \in d(n)} \!\!\! \left| i'_m \right| \le \!\!\! \sum\limits_{m \in d(n)} \!\!\! \frac{|p_m + j q_m|}{\hat{v}^L_m}, \; \forall n \in \mathcal{L}
\end{align}
where $i'_m$ is the net current injection at bus $m$, $d(n)$ is the set of ($n$-inclusive) descendants of $n$ defined in Section~\ref{subsec:distribution_system_notations} and $\hat{v}^L$ is the voltage lower bound by Algorithm~\ref{alg:fixed_point} (so that $\hat{\nu}^L \le \nu^L$ in \eqref{opt:vi_bounds_V}). In \eqref{eqn:In_UB}, the last inequality holds because $|p_m+j q_m| = |v_m| |i'_m|$ and hence $|p_m+j q_m| \ge \hat{v}^L_m |i'_m|$. Given power injection $(p,q)$, \eqref{eqn:In_UB} provides an upper bound on $\ell(p,q)^{1/2}$. Let $\mathcal{I} \subseteq \mathcal{N}$ denote the set of buses with software update failure implying that $(p,q) \in \mathcal{S}(\mathcal{I})$ in \eqref{eqn:SI}. Then, the maximum line current $\sqrt{\ell_n}$ subject to $\mathcal{I}$ is upper bounded by
\begin{align} \label{eqn:iub_KCL}
\begin{split}
    &\max\limits_{(p,q) \in \mathcal{S}(\mathcal{I})} \sum\limits_{m \in d(n)} \frac{|p_m + j q_m|}{\hat{v}^L_m} \\
    = &\sum\limits_{m \in d(n) \setminus \mathcal{I}} \frac{|(\hat{p}^G_m - \hat{p}^L_m) + j (\hat{q}^G_m - \hat{q}^L_m)|}{\hat{v}^L_m} + \sum\limits_{m \in d(n) \cap \mathcal{I}} \frac{{|s^\star|}_m}{\hat{v}^L_m}
\end{split}
\end{align}
where
\begin{align} \label{eqn:sm_opt}
\begin{split}
    {|s^\star|}_m = &\max\limits_{p^G_m, q^G_m} \quad |(\hat{p}^L_m + j \hat{q}^L_m) - (p^G_m + j q^G_m)| \\
    & \text{subject to} \quad p^G_m \ge 0, \;\; |p^G_m + j q^G_m| \le C_m
\end{split}
\end{align}
By Lemma~\ref{thm:abr_max} in Appendix A, ${|s^\star|}_m$ in \eqref{eqn:sm_opt} is equal to
\begin{align} \label{eqn:sm}
\begin{split}
    {|s^\star|}_m = &\max\left\{|\hat{p}^L_m + j (\hat{q}^L_m + C_m)|, |\hat{p}^L_m + j (\hat{q}^L_m - C_m)|, \right.\\
    &\left. |\hat{p}^L_m + j \hat{q}^L_m - C_m (\cos(\theta_{sm}) + j \sin(\theta_{sm})| \right\}
\end{split}
\end{align}
where $\theta_{sm} = \tan^{-1}(\hat{q}^L_m/\hat{p}^L_m)$ if at least one of $\hat{p}^L_m$ and $\hat{q}^L_m$ is nonzero, and $\theta_{sm} = 0$ otherwise. Define
\begin{align} \label{eqn:sm_snorm}
\begin{split}
    & \text{$|s^\star| \in \R^N$ with entries specified in \eqref{eqn:sm}} \\
    & |s^{\text{nom}}| := \left| (\hat{p}^G - \hat{p}^L) + j (\hat{q}^G - \hat{q}^L) \right|
\end{split}
\end{align}
Then, the current upper bound described by \eqref{eqn:iub_KCL}, \eqref{eqn:sm} and \eqref{eqn:sm_snorm} can be written compactly as
\begin{align} \label{eqn:iub_lin}
\begin{split}
    \ell^{1/2} \! &\le \mathcal{D} \Big((\hat{v}^L)^{-1} \! \odot \! |s^{\text{nom}}| \! + \! \text{diag}\big((\hat{v}^L)^{-1} \! \odot (|s^\star| \! - \! |s^{\text{nom}}|)\big) {\bf 1}_{\mathcal{I}}\Big) \\
    &:= I^{\text{iub1}} + W^{\text{iub1}} {\bf 1}_{\mathcal{I}}
\end{split}
\end{align}
In summary, if $\mathcal{I} \subseteq \mathcal{N}$ denotes the set of bus(es) with software update(s), then despite the worst-case update failure, the line current vector is upper bounded by the expression in \eqref{eqn:iub_lin}.

{\bf Linear inequalities guaranteeing current upper limit}: For any net power injections $(p,q)$, \eqref{eqn:DF_l} specifies that
\begin{align*}
    \sqrt{\ell_n} &= (\nu_n)^{-1/2} e_n^\top \left|\mathcal{D} (p + jq) - (\mathcal{D}-I) D_z \ell \right| \\
    &\le (\hat{v}^L_n)^{-1} e_n^\top \big( |\mathcal{D}(p \! + \! jq)| \! + \! |(\mathcal{D} \! - \! I) D_z (\hat{i}^U \! \odot \ell^{1/2}) | \big)
\end{align*}
where $\mathcal{D}$ and $D_z$ are defined in Section~\ref{subsec:distribution_system_notations}, and $\hat{v}^L = (\hat{\nu}^L)^{1/2}$, $\hat{i}^U = (\hat{\ell}^U)^{1/2}$ from Algorithm \ref{alg:fixed_point}. The inequality above is due to the following four facts: (a) $\hat{v}^L \le (\nu(p,q))^{1/2}$, (b) triangular inequality, (c) ${\bf 0} \le \ell(p,q) = \ell^{1/2} \odot \ell^{1/2} \le \hat{i}^U \odot \ell^{1/2}$ and (d) entries of $(\mathcal{D} - I)$, $D_r$, $D_x$ are nonnegative. Let $\mathcal{I} \subseteq \mathcal{N}$ denote the set of buses with software update failure implying that $(p,q) \in \mathcal{S}(\mathcal{I})$ in \eqref{eqn:SI}. Then, for any $n \in \mathcal{N}$
\begin{align} \label{eqn:iub_inexact1}
\begin{split}
    \max\limits_{(p,q) \in \mathcal{S}(\mathcal{I})} \sqrt{\ell_n} \le & \; (\hat{v}^L_n)^{-1} \max\limits_{(p,q) \in \mathcal{S}(\mathcal{I})} e_n^\top \left|\mathcal{D}(p + jq)\right| \\
    & + (\hat{v}^L_n)^{-1} e_n^\top |(\mathcal{D} - I) D_z (\hat{i}^U \odot \ell^{1/2}) |
\end{split}
\end{align}
The two terms in the right-hand-side of \eqref{eqn:iub_inexact1} can be analyzed separately. By Lemma~\ref{thm:Ds_max} in Appendix A the first term is $(\hat{v}^L_n)^{-1} \max\{|a + j (b+r)|, |a + j(b - r)|, |a+jb - r e^{j \theta}|\}$ with $a$, $b$, $r$, $\theta$ defined in \eqref{eqn:DS_max_parameters} in Lemma~\ref{thm:Ds_max}. Due to $\theta$, the expression is a nontrivial function of ${\bf 1}_{\mathcal{I}}$ and hence we consider its simplifying upper bound $(\hat{v}^L_n)^{-1}(|a + jb| + r)$, which is written as
\begin{align} \label{eqn:iub_inexact_max}
\begin{split}
    & (\hat{v}^L_n)^{-1} \Big|e_n^\top \mathcal{D} (\hat{p}^L \! - \! \hat{p}^G \! + \! j( \hat{q}^L \! - \! \hat{q}^G)) \! + \! e_n^\top \mathcal{D} D_{(\hat{p}^G+j \hat{q}^G)} {\bf 1}_{\mathcal{I}}\Big| \\
    &  + (\hat{v}^L_n)^{-1} e_n^\top \mathcal{D} D_C {\bf 1}_{\mathcal{I}}
\end{split}
\end{align}
Because of the absolute value, \eqref{eqn:iub_inexact_max} is not affine with respect to ${\bf 1}_{\mathcal{I}}$. Thus, it is further upper bounded by polygonal approximation. In particular, for any $x+jy \in \mathbb{C}$ it holds that
\begin{align} \label{eqn:polygonal_approximation}
\begin{split}
    & |x+jy| \\
    \le &\frac{1}{\cos(\pi/N')} \max\limits_{1 \le k \le N'} \big\{\cos(\frac{(2k-1)\pi}{N'}) x \!+\! \sin(\frac{(2k-1)\pi}{N'}) y \big\}
\end{split}
\end{align}
for any integer $N' \ge 4$. For example, if $N' = 4$ then
\begin{align*}
    |x+jy| \le \max\{x+y, x-y, -x+y, -x-y\} = |x| + |y|
\end{align*}
Therefore,
\begin{align} \label{eqn:iub_inexact_max_pa}
\begin{split}
    & \text{\eqref{eqn:iub_inexact_max}} \\
    \le & \; (\hat{v}^L_n)^{-1} \frac{1}{\cos(\pi/N')} \! \max\limits_{1 \le k \le N'} \! \Big\{ \cos(\frac{(2k\!-\!1)\pi}{N'}) \big(e_n^\top \mathcal{D} (\hat{p}^L \!-\! \hat{p}^G) \\
    &+ e_n^\top \mathcal{D} D_{\hat{p}^G} {\bf 1}_{\mathcal{I}}\big) +  \sin(\frac{(2k-1)\pi}{N'}) \big(e_n^\top \mathcal{D} (\hat{q}^L - \hat{q}^G) \\
    & + e_n^\top \mathcal{D} D_{\hat{q}^G} {\bf 1}_{\mathcal{I}}\big) \Big\} + (\hat{v}^L_n)^{-1} e_n^\top \mathcal{D} D_C {\bf 1}_{\mathcal{I}}
\end{split}
\end{align}
Now we consider the second term in the right-hand-side of \eqref{eqn:iub_inexact1}. Using \eqref{eqn:iub_lin} and the fact that all entries of $(\mathcal{D}-I)$, $D_r$, $D_x$ are nonnegative, the term can be upper bounded by
\begin{align} \label{eqn:iub_inexact_loss}
    (\hat{v}^L_n)^{-1} \Big|e_n^\top (\mathcal{D} - I) \text{diag}(z \odot \hat{i}^U) (I^{\text{iub1}} + W^{\text{iub1}} {\bf 1}_{\mathcal{I}}) \Big|
\end{align}
Using polygonal approximation in \eqref{eqn:polygonal_approximation} with integer $N'' \ge 4$, it holds that
\begin{align} \label{eqn:iub_inexact_loss_pa}
    \begin{split}
        & \text{\eqref{eqn:iub_inexact_loss}} \\
        \le & \; (\hat{v}^L_n)^{-1} \frac{1}{\cos\left(\pi/N''\right)} \max\limits_{1 \le k \le N''} \Big\{ \cos(\frac{(2k\!-\!1)\pi}{N''}) e_n^\top (\mathcal{D} - I) \\
        & \times \text{diag}(r \odot \hat{i}^U) (I^{\text{iub1}} + W^{\text{iub1}} {\bf 1}_{\mathcal{I}}) + \sin(\frac{(2k-1)\pi}{N''}) \\
        & \times e_n^\top (\mathcal{D} - I) \text{diag}(x \odot \hat{i}^U) (I^{\text{iub1}} + W^{\text{iub1}} {\bf 1}_{\mathcal{I}}) \Big\}
    \end{split}
\end{align}
Hence, the sum of the RHS of \eqref{eqn:iub_inexact_max_pa} and \eqref{eqn:iub_inexact_loss_pa} is an upper bound of $\sqrt{\ell_n(p,q)}$ for all $(p,q) \in \mathcal{S}(\mathcal{I})$. Thus, $\forall (p,q) \in \mathcal{S}(\mathcal{I})$
\begin{align} \label{eqn:iub_con_n}
    \text{RHS \eqref{eqn:iub_inexact_max_pa}} + \text{RHS \eqref{eqn:iub_inexact_loss_pa}} \le (\bar{\ell}_n)^{1/2}
    \! \implies \! \sqrt{\ell_n} \le (\bar{\ell}_n)^{1/2}
\end{align}
Utilizing the fact that for any $x \in \R^{N'}$, $y \in \R^{N''}$, $z \in \R$
\begin{align*}
    \max\limits_i x_i + \max\limits_j y_j \le z \iff {\bf 1}_{N''} \otimes x + y \otimes {\bf 1}_{N'} \le z {\bf 1}_{N' N''}
\end{align*}
with ${\bf 1}_{N'}$, ${\bf 1}_{N''}$, ${\bf 1}_{N' N''}$ denoting all-one vectors of appropriate dimensions, the condition in \eqref{eqn:iub_con_n} can be ``vectorized'' as a set of linear constraints with respect to ${\bf 1}_{\mathcal{I}}$ to enforce $\sqrt{\ell_n} \le (\bar{\ell}_n)^{1/2}$ for a specific $n \in \mathcal{N}$. By ``vectorizing'' the conditions in \eqref{eqn:iub_con_n} for all $n$, the desired linear constraints can be summarized as
\begin{equation} \label{eqn:iub_con}
    I^{\text{iub2}} + W^{\text{iub2}} {\bf 1}_{\mathcal{I}} \le {\bf 1}_{N''} \otimes ({\bf 1}_{N'} \otimes (\bar{\ell})^{1/2})
\end{equation}
where
\begin{align} \label{eqn:WI_iub2}
    \begin{split}
        I^{\text{iub2}} &= {\bf 1}_{N''} \otimes \Big(c' \otimes \big((\hat{v}^L)^{-1} \odot \mathcal{D} (\hat{p}^L - \hat{p}^G) \big) \\ 
        &+ s' \otimes \big((\hat{v}^L)^{-1} \odot \mathcal{D} (\hat{q}^L - \hat{q}^G) \big) \Big) \\
        &+ c'' \! \otimes \! {\bf 1}_{N'} \otimes \big({(\hat{v}^L)}^{-1} \odot (\mathcal{D} \! - \! I) \text{diag}(r \odot \hat{i}^U) I^{\text{iub1}} \big) \\
        &+ s'' \! \otimes \! {\bf 1}_{N'} \otimes \big({(\hat{v}^L)}^{-1} \odot (\mathcal{D} \! - \! I) \text{diag}(x \odot \hat{i}^U) I^{\text{iub1}} \big) \\
        W^{\text{iub2}} &= {\bf 1}_{N''} \otimes \Big(c' \otimes D_{\hat{v}^L}^{-1} \mathcal{D} D_{\hat{p}^G} + s' \otimes D_{\hat{v}^L}^{-1} \mathcal{D} D_{\hat{q}^G} \\
        &+ {\bf 1}_{N'} \otimes D_{\hat{v}^L}^{-1} \mathcal{D} D_C \Big)  \\
        &+ c'' \otimes {\bf 1}_{N'} \otimes \big(D_{\hat{v}^L}^{-1} (\mathcal{D} - I) \text{diag}(r \odot \hat{i}^U) W^{\text{iub1}}\big) \\
        &+ s'' \otimes {\bf 1}_{N'} \otimes \big(D_{\hat{v}^L}^{-1} (\mathcal{D} - I) \text{diag}(x \odot \hat{i}^U) W^{\text{iub1}}\big)
    \end{split}
\end{align}
and $c', s' \in \R^{N'}$, $c'', s'' \in \R^{N''}$ with entries defined by
\begin{align} \label{eqn:cs_vectors}
    \begin{split}
        c'_k &= \frac{\cos((2k-1) \pi/N')}{\cos(\pi/N')}, \quad \forall k \in \{1, \ldots, N'\} \\
        s'_k &= \frac{\sin((2k-1) \pi/N')}{\cos(\pi/N')}, \quad \forall k \in \{1, \ldots, N'\} \\
        c''_k &= \frac{\cos((2k-1) \pi/N'')}{\cos(\pi/N'')}, \quad \forall k \in \{1, \ldots, N''\} \\
        s''_k &= \frac{\sin((2k-1) \pi/N'')}{\cos(\pi/N'')}, \quad \forall k \in \{1, \ldots, N''\}
    \end{split}
\end{align}
In summary, if $\mathcal{I} \subseteq \mathcal{N}$ denotes the set of bus(es) with software update(s), then despite all possible worst-case update failure, if \eqref{eqn:iub_con} holds then the (squared) line currents are no greater than the upper limits specified by $\bar{\ell}$.

We note that \eqref{eqn:iub_lin} can also be used to construct line current upper limit constraints by imposing $I^{\text{iub1}} + W^{\text{iub1}} {\bf 1}_{\mathcal{I}} \le (\bar{\ell})^{1/2}$. This is in fact adopted by our previous work in \cite{sou2022resilient}, requiring $N$ instead of $N N' N''$ linear constraints in \eqref{eqn:iub_con} used in this paper. However, \eqref{eqn:iub_lin} is derived by replacing the current phasors with their magnitudes in \eqref{eqn:In_UB}. This tends to be conservative especially when there is significant reverse power flows in the system. Our numerical experiment (not shown in this paper) indeed indicates that \eqref{eqn:iub_con} are generally less conservative than \eqref{eqn:iub_lin}. In addition, the total time to construct the constraints in \eqref{eqn:iub_con} and to solve the corresponding rollout problem remains modest (see Section~\ref{sec:demonstration} for more detail). This motivates our proposed current upper limit constraints using \eqref{eqn:iub_con}. Nevertheless, \eqref{eqn:iub_lin} retains its distinct advantage that it is a upper bound of $(\ell)^{1/2}$ linearly dependent on ${\bf 1}_{\mathcal{I}}$. Usually for $\mathcal{I} \subset \mathcal{N}$, $I^{\text{iub1}} + W^{\text{iub1}} {\bf 1}_{\mathcal{I}} \le \hat{i}^U$ because $\hat{i}^U$ assumes update failure at all buses in $\mathcal{N}$. This observation is utilized to obtain a less conservative second term in \eqref{eqn:iub_inexact1} (as in \eqref{eqn:iub_inexact_loss}).

\subsection{Voltage Upper Limit Constraint} \label{sec:v_UB}
Note that $M \ell \le 0$ in \eqref{eqn:DF_nu} because all entries of $M$ in \eqref{eqn:RXM} are non-positive. Therefore, for any $(p,q)$, $\nu(p,q) \le \nu_0 {\bf 1} + 2 R p + 2 X q$. Let $\mathcal{I} \subseteq \mathcal{N}$ denote the set of buses with software updates\ implying that $(p,q) \in \mathcal{S}(\mathcal{I})$ in \eqref{eqn:SI}. Then, for any $n \in \mathcal{N}$,
\begin{align} \label{eqn:vub0}
\begin{split}
    & \max\limits_{(p,q) \in \mathcal{S}(\mathcal{I})} \nu_n \\
    \le & \max\limits_{(p,q) \in \mathcal{S}(\mathcal{I})} \big\{ \nu_0 + 2 e_n^\top R p + 2 e_n^\top X q \big\} \\
    = & \; e_n^\top \Big( \nu_0 {\bf 1} + 2 R (\hat{p}^G - \hat{p}^L) + 2 X (\hat{q}^G - \hat{q}^L) \Big) \\
    & + \sum\limits_{i \in \mathcal{I}} \max\limits_{p^G_i \ge 0, \; q^G_i} \Big\{ 2 R_{ni} p^G_i + 2 X_{ni} q^G_i \mid |p^G_i + j q^G_i| \le C_i \Big\} \\
    & - \sum\limits_{i \in \mathcal{I}} (2 R_{ni} \hat{p}^G_i + 2 X_{ni} \hat{q}^G_i)
\end{split}
\end{align}
where $R_{ni}$ and $X_{ni}$ denote the $(n,i)$ entry of $R$ and $X$, respectively. Since $R_{ni} \ge 0$, $X_{ni} \ge 0$, it holds in \eqref{eqn:vub0} that
\begin{align} \label{eqn:vub_max}
    \begin{split}
        & \max\limits_{p^G_i \ge 0, \; q^G_i} \Big\{ 2 R_{ni} p^G_i + 2 X_{ni} q^G_i \mid |p^G_i + j q^G_i| \le C_i \Big\} \\
        = & \; 2 C_i \sqrt{R_{ni}^2 + X_{ni}^2} \\
        := & \; 2 C_i Z_{ni}
    \end{split}
\end{align}
with
\begin{align} \label{eqn:Z}
    Z \in \R^{N \times N}, \;\; Z_{ni} := \sqrt{R_{ni}^2 + X_{ni}^2}, \;\; \forall (n,i)
\end{align}
Note that in \eqref{eqn:vub0} it holds (for the last term) that
\begin{equation} \label{eqn:RpXq}
\sum\limits_{i \in \mathcal{I}} (2 R_{ni} \hat{p}^G_i + 2 X_{ni} \hat{q}^G_i) = e_n^\top (2 R D_{\hat{p}^G} + 2 X D_{\hat{q}^G}) {\bf 1}_{\mathcal{I}}
\end{equation}
Thus, \eqref{eqn:vub_max} and \eqref{eqn:RpXq} imply that the inequalities in \eqref{eqn:vub0} can be assembled for all $n$ as
\begin{align} \label{eqn:vub}
    \nu(p,q) \le \hat{\nu}^{\text{nom}} + W^{\text{vub}} {\bf 1}_{\mathcal{I}}, \quad \forall (p,q) \in \mathcal{S}(\mathcal{I})
\end{align}
where
\begin{align} \label{eqn:vub_vW}
    \begin{split}
        \hat{\nu}^{\text{nom}} &:= \nu_0 {\bf 1} + 2 R (\hat{p}^G - \hat{p}^L) + 2 X (\hat{q}^G - \hat{q}^L) \\
        W^{\text{vub}} &:= 2 Z D_C - (2 R D_{\hat{p}^G} + 2 X D_{\hat{q}^G})
    \end{split}
\end{align}
In summary, if $\mathcal{I}$ denotes the set of buses with software updates the maximum squared voltages are upper bounded by $\hat{\nu}^{\text{nom}} + W^{\text{vub}} {\bf 1}_{\mathcal{I}}$ as specified in \eqref{eqn:vub}. Therefore, if the constraints 
\begin{equation} \label{eqn:vub_con}
\hat{\nu}^{\text{nom}} + W^{\text{vub}} {\bf 1}_{\mathcal{I}} \le \bar{\nu}
\end{equation}
are imposed on ${\bf 1}_{\mathcal{I}}$, then the worst-case possible squared voltages due to software update failure associated with inverters at buses in $\mathcal{I}$ cannot exceed $\bar{\nu}$.

We remark that $\hat{\nu}^{\text{nom}}$ in \eqref{eqn:vub_vW} is the vector of nominal squared voltage (i.e., no software update failure) of the \emph{linearized} DistFlow equations (i.e., \eqref{eqn:DF_nu} with the loss term $M \ell$ ignored).


\subsection{Voltage Lower Limit Constraint} \label{sec:v_LB}
Let $\mathcal{I} \subseteq \mathcal{N}$ denotes the set of bus(es) with software update(s) implying that $(p,q) \in \mathcal{S}(\mathcal{I})$. Then, for any $n \in \mathcal{N}$, \eqref{eqn:DF_nu} specifies that
\begin{align*}
\begin{split}
    \nu_n &= \nu_0 + 2 e_n^\top Rp + 2 e_n^\top X q + e_n^\top M \ell \\
    &\ge \nu_0 + e_n^\top \Big( 2 Rp + 2 X q + M \big(\hat{i}^U \odot (I^{\text{iub1}} + W^{\text{iub1}} {\bf 1}_{\mathcal{I}}) \big) \Big)
\end{split}
\end{align*}
where the inequality holds because of the following four facts: (a) all entries of $M$ are non-positive, (b) $\ell = \ell^{1/2} \odot \ell^{1/2}$, (c) for all $(p,q)$, it holds that $\ell^{1/2} \le \hat{i}^U$ by \eqref{opt:vi_bounds_l}, (d) for all $(p,q) \in \mathcal{S}(\mathcal{I})$, it holds that $\ell^{1/2} \le I^{\text{iub1}} + W^{\text{iub1}} {\bf 1}_{\mathcal{I}}$. Therefore, the minimum $\nu_n$ due to $(p,q) \in \mathcal{S}(\mathcal{I})$ satisfies
\begin{align} \label{eqn:vlb1}
    \begin{split}
        & \min\limits_{(p,q) \in \mathcal{S}(\mathcal{I})} \nu_n \\ 
        \ge & \; \nu_0 + \min\limits_{(p,q) \in \mathcal{S}(\mathcal{I})} \big\{ 2 e_n^\top R p + 2 e_n^\top X q \big\} \\
        & + e_n^\top M \big(\hat{i}^U \odot (I^{\text{iub1}} + W^{\text{iub1}} {\bf 1}_{\mathcal{I}}) \big) \\
        = & \; e_n^\top \Big( \nu_0 {\bf 1}_{\mathcal{I}} + 2 R (\hat{p}^G - \hat{p}^L) + 2 X (\hat{q}^G - \hat{q}^L) \Big) \\
        &+ \sum\limits_{i \in \mathcal{I}} \min\limits_{p^G_i \ge 0, q^G_i} \Big\{2 R_{ni} p^G_i + 2 X_{ni} q^G_i \mid |p^G_i + j q^G_i| \le C_i \Big\} \\
        &- \sum\limits_{i \in \mathcal{I}} (2 R_{ni} \hat{p}^G_i \! + \! 2 X_{ni} \hat{q}^G_i) \! + \! e_n^\top M \big(\hat{i}^U \! \odot \!(I^{\text{iub1}} + W^{\text{iub1}} {\bf 1}_{\mathcal{I}}) \big)
    \end{split}
\end{align}
Since $R_{ni} \ge 0$, $X_{ni} \ge 0$, in \eqref{eqn:vlb1} it holds that
\begin{align} \label{eqn:vlb_min}
    \begin{split}
        & \min\limits_{p^G_i \ge 0, q^G_i} \Big\{2 R_{ni} p^G_i + 2 X_{ni} q^G_i \mid |p^G_i + j q^G_i| \le C_i \Big\} \\
        = & - \sum\limits_{i \in \mathcal{I}} 2 X_{ni} C_i \\
        = & - 2 e_n^\top X D_C {\bf 1}_{\mathcal{I}}
    \end{split}
\end{align}
Thus, \eqref{eqn:RpXq}, \eqref{eqn:vub_vW} and \eqref{eqn:vlb_min} imply that the inequalities in \eqref{eqn:vlb1} can be assembled for all $n$ as
\begin{equation} \label{eqn:vlb}
    \nu(p,q) \ge \hat{\nu}^{\text{nom}} + M (\hat{i}^U \odot I^{\text{iub1}}) - W^{\text{vlb}} {\bf 1}_{\mathcal{I}}, \;\; \forall (p,q) \in \mathcal{S}(\mathcal{I})
\end{equation}
where
\begin{equation} \label{eqn:vlb_W}
    W^{\text{vlb}} := 2 X D_C + 2 R D_{\hat{p}^G} + 2 X D_{\hat{q}^G} - M D_{\hat{i}^U} W^{\text{iub1}}
\end{equation}
In summary, if $\mathcal{I}$ denotes the set of buses with software updates then the minimum squared voltages are lower bounded by $\hat{\nu}^{\text{nom}} + M (\hat{i}^U \odot I^{\text{iub1}}) - W^{\text{vlb}} {\bf 1}_{\mathcal{I}}$ as in \eqref{eqn:vlb}. Therefore, if the constraints
\begin{equation} \label{eqn:vlb_con}
\hat{\nu}^{\text{nom}} + M (\hat{i}^U \odot I^{\text{iub1}}) - W^{\text{vlb}} {\bf 1}_{\mathcal{I}} \ge \underline{\nu}
\end{equation}
are imposed on ${\bf 1}_{\mathcal{I}}$, then the worst-case squared voltages due to software update failure associated with inverters at buses in $\mathcal{I}$ cannot be lower than $\underline{\nu}$.

\subsection{Software Update Rollout Problem Formulation} \label{sec:formulation}
The decision variables of the rollout problem are the software update instants $t_i \ge 0$ for bus $i \in \mathcal{N}$. Since the updates should be scheduled to cause the least interruption to normal operation, the makespan (i.e., the time when the first update starts to the time when the last update finishes) should be minimized. Furthermore, the safety constraints should be satisfied. For given $t_1, t_2, \ldots, t_N$ and $t$ we define
\begin{equation} \label{eqn:It}
    \mathcal{I}(t) := \{i \in \mathcal{N} \mid t \in \begin{bmatrix} t_i, t_i + \Delta \tau \end{bmatrix} \}
\end{equation}
as the set of buses such that time interval $\begin{bmatrix} t_i, t_i + \Delta \tau \end{bmatrix}$ contains $t$. Recall that $\Delta \tau$ is the fault clearing time, so $\mathcal{I}(t)$ is the set of buses whose adverse effect of software update failure can be seen at time $t$ (for given $t_1, t_2, \ldots$). Then, to maintain the grid code in \eqref{eqn:lv_bounds} in the worst-case scenario (all scheduled updates fail and the corresponding nodes suffer the worst possible power injections), constraints \eqref{eqn:iub_con}, \eqref{eqn:vub_con} and \eqref{eqn:vlb_con} should hold for $\mathcal{I}(t)$ at all $t$. Hence, the rollout problem is summarized as
\begin{align} \label{opt:rollout}
    \begin{split}
        \underset{t_i}{\text{minimize}} & \quad \underset{i \in \mathcal{N}}{\max} \; t_i \\
        \text{subject to} & \quad {\bf H} {\bf 1}_{\mathcal{I}(t)} \le {\bf b}, \quad \forall t \\
        & \quad \text{$\mathcal{I}(t)$ satisfies \eqref{eqn:It}}
    \end{split}
\end{align}
where
\begin{equation} \label{eqn:Hb}
    \begin{split}
        \!\! {\bf H} := \begin{bmatrix}
        W^{\text{iub2}} \\ W^{\text{vub}} \\ W^{\text{vlb}}
        \end{bmatrix}, \;
        \mathbf{b} := \begin{bmatrix}
        {\bf 1}_{N''} \otimes ({\bf 1}_{N'} \otimes (\bar{\ell})^{1/2}) - I^{\text{iub2}} \\ \overline{\nu} - \hat{\nu}^{\text{nom}} \\ \hat{\nu}^{\text{nom}} + M (\hat{i}^U \odot I^{\text{iub1}}) - \underline{\nu}
        \end{bmatrix}
    \end{split}
\end{equation}
Due to the indicator function $\mathcal{I}(t)$, problem \eqref{opt:rollout} is nonlinear and non-convex. However, as shown in \cite[Proposition 1]{de2020minimum}, $t_i$ in \eqref{opt:rollout} can in fact be restricted to nonnegative integer multiples of $\Delta \tau$. Then, \eqref{opt:rollout} can be reinterpreted as follows. The time axis can be divided into intervals called time slots of the form $[k \Delta \tau, (k+1) \Delta \tau]$ with $0 \le k \le N-1$. The decision to make is the assignment of updates to the time slots, so that each update is assigned to exactly one time slot. For the shortest update schedule, the number of time slots used is minimized. Further, instead of imposing the constraint in \eqref{opt:rollout} for all $t$, it suffices to consider only integer multiples of $\Delta \tau$. By defining the ``discrete-time'' version of $\mathcal{I}$ in \eqref{eqn:It} with time slot index starting from $j = 1$
\begin{displaymath}
\mathcal{I}_j := \mathcal{I}\big((j-1) \Delta \tau\big), \quad j = 1, 2, \ldots, N
\end{displaymath}
problem \eqref{opt:rollout} can be reformulated as follows:
\begin{align} \label{opt:bin_packing}
    \begin{split}
        \text{minimize} & \quad T \\
        \text{subject to} & \quad \text{$\mathcal{I}_1, \mathcal{I}_2, \ldots, \mathcal{I}_T$ partition $\mathcal{N}$} \\
        & \quad {\bf H} {\bf 1}_{\mathcal{I}_j} \le {\bf b}, \quad j \in \{1, \ldots, T\}
    \end{split}
\end{align}
with ${\bf H}$ and ${\bf b}$ defined in \eqref{eqn:Hb}. Problem~\eqref{opt:bin_packing} is known as the \emph{vector bin packing problem}. It can be modeled as an integer program (e.g., \cite{de2020minimum}) or approximately solved using greedy algorithms to be detailed in the next section.

\section{Solving Software Update Rollout Problem} \label{sec:solution}
To quickly solve large instances of \eqref{opt:bin_packing}, we adopt the best-fit-decreasing greedy algorithm. First we check if it is possible to let each time slot hold exactly one update, since the instance is feasible if and only if this is possible. If feasible, the updates are sorted by their ``sizes'' specified by the row vector ${\bf 1}^\top {\bf H}$ (large to small). Then, the sorted updates are considered one-by-one. For each update, we first try to assign it to the time slot with the minimum ``sum-up residual capacity'' (to be defined in Algorithm~\ref{alg:BFD}), among the used slots with some update already assigned. Only if no used slot can hold the update, will a new slot be appended to the schedule.

Algorithm~\ref{alg:BFD} summaries our customized best-fit-decreasing algorithm.
\begin{algorithm}
\caption{Best-fit-decreasing algorithm to solve \eqref{opt:bin_packing}}
\label{alg:BFD}
\begin{algorithmic}[1]
\Require ${\bf H} \in \R^{K \times N}$, ${\bf b} \in \R^K$
\Ensure ${\bf x} \in \{0,1\}^{N \times N}$, ${\bf y} \in \{0,1\}^N$ if instance is feasible; ${\bf x}_{ij} = 1$ iff update $i$ is assigned to time slot $j$, and ${\bf y}_j = 1$ iff time slot $j$ is needed in the rollout schedule
\If{${\bf b} \ngeq {\bf H}(:,i)$ for any $i \in \mathcal{N}$} \Comment{feasibility check}
\State report infeasibility, QUIT \EndIf
\State Define ${\bf s} = {\bf 1}^\top {\bf H}$ \Comment{``size'' of each update}
\State Define update indices ${\bf u}_1, \ldots, {\bf u}_N : {\bf s}({\bf u}_1) \ge {\bf s}({\bf u}_2) \ge \ldots$

\State Initialize ${\bf x} \gets {\bf 0}_{N \times N}$, ${\bf y} \gets {\bf 0}_N$
\State Initialize $N_y \gets 0$ \Comment{number of time slots used}
\State Initialize ${\bf C} \gets {\bf 0}_{K \times N}$ (residual capacity vectors for all time slots), ${\bf c} \gets {\bf 1}^\top {\bf b} $ (sum-up residual capacity)
\State Initialize time slot indices ${\bf s}_1 \gets 1$, \ldots, ${\bf s}_N \gets N$

\For{$i = 1, \ldots, N$} \Comment{loop over sorted updates}
\For{$j = 1, \ldots, N$} \Comment{loop over sorted time slots}
\If{${\bf C}(:,{\bf s}_j) \ge {\bf H}(:,{\bf u}_i)$} \Comment{check used slot}
   \State $s^* \gets {\bf s}_j$, quit for-loop over $j$
\EndIf
\If{$j > N_y$} \Comment{open new slot}
\State $s^* \gets j$, $N_y +\!+$, ${\bf C}(:,N_y) \gets {\bf b}$, ${\bf y}(N_y) \gets 1$
\State quit for-loop over $j$
\EndIf
\EndFor

\State Update ${\bf x}({\bf u}_i, s^*) \gets 1$
\State Update ${\bf C}(:,s^*) \minuseq H(:,{\bf u}_i)$, ${\bf c}(s^*) \minuseq {\bf s}({\bf u}_i)$

\While{$j \ge 2$} \Comment{update used time slots ordering}
\If{${\bf c}({\bf s}_j) < {\bf c}({\bf s}_{j-1})$}
\State Swap ${\bf s}_j$ with ${\bf s}_{j-1}$
\EndIf
\State $j -\!-$
\EndWhile
\EndFor
\end{algorithmic}
\end{algorithm}
Assume that the instance is feasible. The assignment ${\bf C}(:,N_y) \gets {\bf b}$ in line 16 specifies that whenever a new time slot is used, its capacity vector is initialized to be ${\bf b}$. Whenever an update is assigned to the time slot, the modification ${\bf C}(:,s^*) \minuseq H(:,{\bf u}_i)$ in line 21 ensures that the residual capacity vector of the slot is properly deducted. Then the condition ${\bf C}(:,{\bf s}_j) \ge {\bf H}(:,{\bf u}_i)$ in line 12 specifies that the time slot can accept an update only if its residual capacity vector remains nonnegative after the assignment. Together, the assignment returned by the algorithm satisfies the second constraints in \eqref{opt:bin_packing}. In addition, the for-loop over $i$ implies that each update must be assigned to exactly one time slot (a new slot if necessary). Hence, the first constraint of \eqref{opt:bin_packing} is satisfied by the algorithm output. Therefore, the algorithm returns a feasible rollout whenever the instance is feasible. Though the returned rollout is not necessarily optimal, it typically requires very few time slots (see Table~\ref{tab:statistics} in Section~\ref{sec:demonstration} for detail).

The main loop over $i$ is executed $N$ times. Each iteration of the inner loop over $j$ requires $O(K)$ operations. The remaining steps in an iteration of the main loop require $O(N+K)$ operations. Hence, the time-complexity of Algorithm~\ref{alg:BFD} is $O(N^2 K)$. In Section~\ref{sec:demonstration}, we demonstrate that Algorithm~\ref{alg:BFD} as stated can be executed very efficiently even in MATLAB. For example, running Algorithm~\ref{alg:BFD} for the 10,476-bus instance from \cite{REDS_KA} requires less than 2 seconds. The overall software rollout scheduling procedure is summarized in Algorithm~\ref{alg:procedure_summary}.
\begin{algorithm}
    \caption{Software update rollout scheduling procedure}
    \label{alg:procedure_summary}
    \begin{algorithmic}[1]
        \Require (a) DistFlow equation \eqref{eqn:DF}, (b) load $(\hat{p}^L, \hat{q}^L)$ and nominal generation $(\hat{p}^G, \hat{q}^G)$, (c) inverter capacity $C$ in \eqref{eqn:inverter_constraints} and (d) safety limits $\underline{\nu}$, $\overline{\nu}$ and $\overline{\ell}$ in \eqref{eqn:lv_bounds}
        \State Compute universal voltage lower bound $\hat{\nu}^L = (\hat{v}^L)^2$ and current upper bound $\hat{\ell}^U = (\hat{i}^U)^2$ according to Algorithm~\ref{alg:fixed_point}
        \State Compute $W^{\text{iub2}}$ and $I^{\text{iub2}}$ according to \eqref{eqn:WI_iub2}, \eqref{eqn:cs_vectors} and \eqref{eqn:iub_lin}
        \State Compute $W^{\text{vub}}$ and $\hat{\nu}^{\text{nom}}$ according to \eqref{eqn:vub_vW}, \eqref{eqn:RXM} and \eqref{eqn:Z}
        \State Compute $W^{\text{vlb}}$, $M$ and $I^{\text{iub1}}$ using \eqref{eqn:vlb_W}, \eqref{eqn:RXM} and \eqref{eqn:iub_lin}
        \State Assemble matrix ${\bf H}$ and vector ${\bf b}$ according to \eqref{eqn:Hb}
        \State Apply Algorithm~\ref{alg:BFD} to decide if instance $({\bf H}, {\bf b})$ is feasible. If feasible, rollout schedule is defined by output ${\bf x}$
    \end{algorithmic}
\end{algorithm}

\section{Numerical Demonstration} \label{sec:demonstration}

We demonstrate the practical performance of the proposed software update rollout scheduling procedure with distribution system benchmarks. All computation is performed on a Mac Studio (2022) with M1 Ultra CPU and 128 GB of RAM, running MATLAB through the Rosetta 2 translation environment.

\subsection{Distribution System Benchmarks}
The benchmarks considered in this paper include: CIGRE 44-bus low-voltage (LV) radial network from \cite{pandapower.2018}, radial networks from Matpower 7.1 \cite{zimmerman2010matpower} and the REDS repository \cite{REDS_KA}. The CIGRE network contains one substation busbar, three feeder heads and 40 non-reference buses. In our study, the substation and the feeder heads are merged into one slack bus with squared voltage fixed at $\nu_0 = 1$ pu. The descriptions (topology, line parameters, load, etc.) of all other benchmarks can be found in \cite{zimmerman2010matpower} and \cite{REDS_KA}. We assume that for each benchmark all non-reference buses are equipped with smart inverters with the same (benchmark dependent) power rating.

\subsection{Quality of Universal Voltage and Current Bounds} \label{subsec:bound_quality}
We numerically demonstrate the quality of the voltage lower bound $\hat{\nu}^L$ and current upper bound $\hat{\ell}^U$ obtained by Algorithm~\ref{alg:fixed_point}. Due to \eqref{opt:vi_bounds} and Proposition~\ref{thm:fixed_point}, for any $(p,q)$ it holds that $\nu(p,q) \ge \nu^L \ge \hat{\nu}^L$ and $\ell(p,q) \le \ell^U \le \hat{\ell}^U$. Hence, if the difference $\nu(p,q) - \hat{\nu}^L$ is small then the error $\nu^L - \hat{\nu}^L$ must be small. A similar relation holds for $\ell(p,q)$, $\ell^U$ and $\hat{\ell}^U$. A suitable choice of $(p,q)$ is $(\tilde{p},\tilde{q}) = (-\hat{p}^L, -\hat{q}^L - C)$ because $(\tilde{p},\tilde{q})$ minimizes $\nu$ in \eqref{eqn:DF_nu} if the loss term $M \ell$ is ignored (also making $\ell$ large in in \eqref{eqn:DF_l} consequently). For the CIGRE LV benchmark we compare $\hat{\nu}^L$ and $\nu(\tilde{p},\tilde{q})$ (resp.~$\hat{\ell}^U$ and $\ell(\tilde{p},\tilde{q})$) in two cases. In both cases the total inverter power rating is the total active power load. However, for the first case the default load is considered resulting in typical forward power flows. On the other hand, for the second case ``background'' distributed generation is added so that the total net active power load is zero. This is a situation with significant distributed generation and reverse power flows. Fig.~\ref{fig:CIGRE_bounds} indicates that $\hat{\nu}^L$ and $\hat{\ell}^U$ obtained by Algorithm~\ref{alg:fixed_point} are of acceptable quality for the CIGRE benchmark.
\begin{figure}[!tb]
    \centering
    \includegraphics[width=88mm]{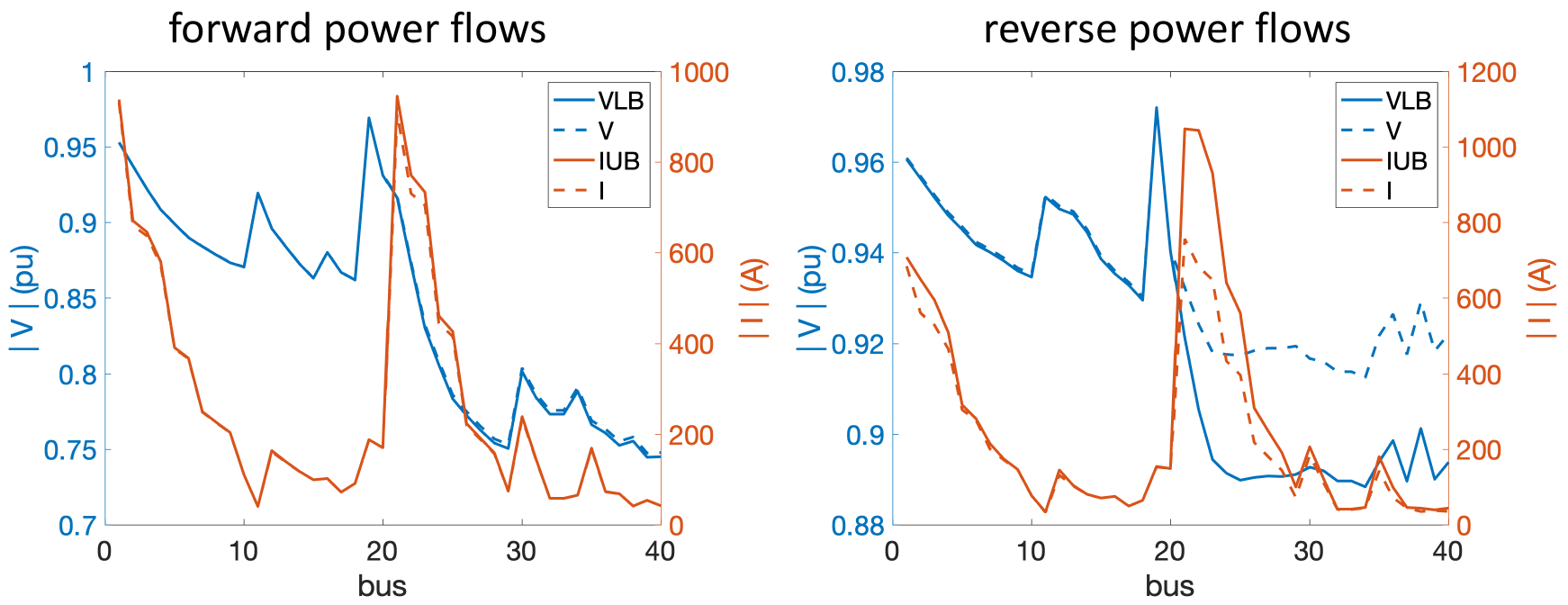}
    \caption{Quality demonstration of universal voltage and current bounds with CIGRE LV network. Solid blue lines are $\hat{\nu}^L$ obtained by Algorithm~\ref{alg:fixed_point}. Dotted blue lines are $\nu(\tilde{p},\tilde{q})$, upper bounds of $\hat{\nu}^L$. Solid orange lines are $\hat{\ell}^U$ obtained by Algorithm~\ref{alg:fixed_point}. Dotted orange lines are $\ell(\tilde{p},\tilde{q})$, lower bounds of $\hat{\ell}^U$. Left: typical load with forward power flows. The average relative error in voltage is $0.17 \%$ and the relative error in current is $1.1 \%$. Right: zero net active power load with significant reverse power flows. The average relative error in voltage is $1.4 \%$ and the relative error in current is $15 \%$.}
    \label{fig:CIGRE_bounds}
\end{figure}
The aforementioned experiment is repeated for selected benchmarks from \cite{zimmerman2010matpower} and \cite{REDS_KA}. Table~\ref{tab:quality} suggests similar conclusions. Except for the approximately 10\% average relative error for current upper bound in scenarios with significant reverse power flows, generally the error is negligible. We note that, in the experiment above, on average Algorithm~\ref{alg:fixed_point} converges in 6.5 fixed point iterations and the maximum number of iterations is 9 over all cases. The relative tolerance in Algorithm~\ref{alg:fixed_point} is $10^{-6}$.
\begin{table}[!htb]
    \centering
    \caption{Average (over buses or lines) relative error of universal voltage lower bounds and current upper bounds}
    \label{tab:quality}
    \begin{tabular}{|c|c|c|c|c|}
        \hline
        \multirow{2}{*}{case name} & \multicolumn{2}{|c|}{forward power flows} & \multicolumn{2}{|c|}{reverse power flows} \\
        \cline{2-5}
        & rel err $\nu$ & rel err $\ell$ & rel err $\nu$ & rel err $\ell$ \\
        \hline
        Matpower 33bw & 0.006\% & 0.21\% & 0.059\% & 5.6\% \\
        \hline
        Matpower 85 & 0.050\% & 0.38\% & 0.064\% & 1.3\% \\
        \hline
        Matpower 118zh & 0.0050\% & 0.22\% & 0.017\% & 11\% \\
        \hline
        Matpower 141 & 0.0018\% & 0.047\% & 0.023\% & 3.9\% \\
        \hline
        REDS 135+8 & $\approx 0$ & 0.022\% & 0.0070\% & 7.3\% \\
        \hline
        REDS 201+3 & $\approx 0$ & 0.012\% & $\approx 0$ & 4.5\% \\
        \hline
        REDS 873+7 & $\approx 0$ & 0.011\% & $\approx 0$ & 4.0\% \\
        \hline
        REDS 10476+84 & $\approx 0$ & 0.010\% & $\approx 0$ & 4.2\% \\
        \hline
    \end{tabular}
\end{table}


\subsection{CIGRE 44-bus LV Distribution Grid -- Heavy Load} \label{subsec:CIGRE_HL}
In this case study, we assume that the power rating of each smart inverter is $0.7/N = 0.015$ pu (the system's total inverter rating is 0.7 pu). The (linear) voltage limits in \eqref{eqn:lv_bounds} are $0.85$ pu and $1.1$ pu for all buses. For each line, the (linear) current upper limit is the maximum of 600 A and 1.5 times the nominal current of the line in the default load setting.

To simulate the scenario with heavy load, the nominal distributed generation is zero (i.e., $\hat{p}^G = 0$ and $\hat{q}^G = 0$) and the load is 130\% of the default load (both for $\hat{p}^L$ and $\hat{q}^L$) so that the total active power load is 0.89 pu. With all relevant data specified, the rollout problem in \eqref{opt:bin_packing} is set up and solved using Algorithm~\ref{alg:BFD}. The Gantt diagram of the obtained update schedule is shown in Fig.~\ref{fig:CIGRE_HL} (left), utilizing nine time slots.
\begin{figure}[!tb]
    \centering
    \includegraphics[width=88mm]{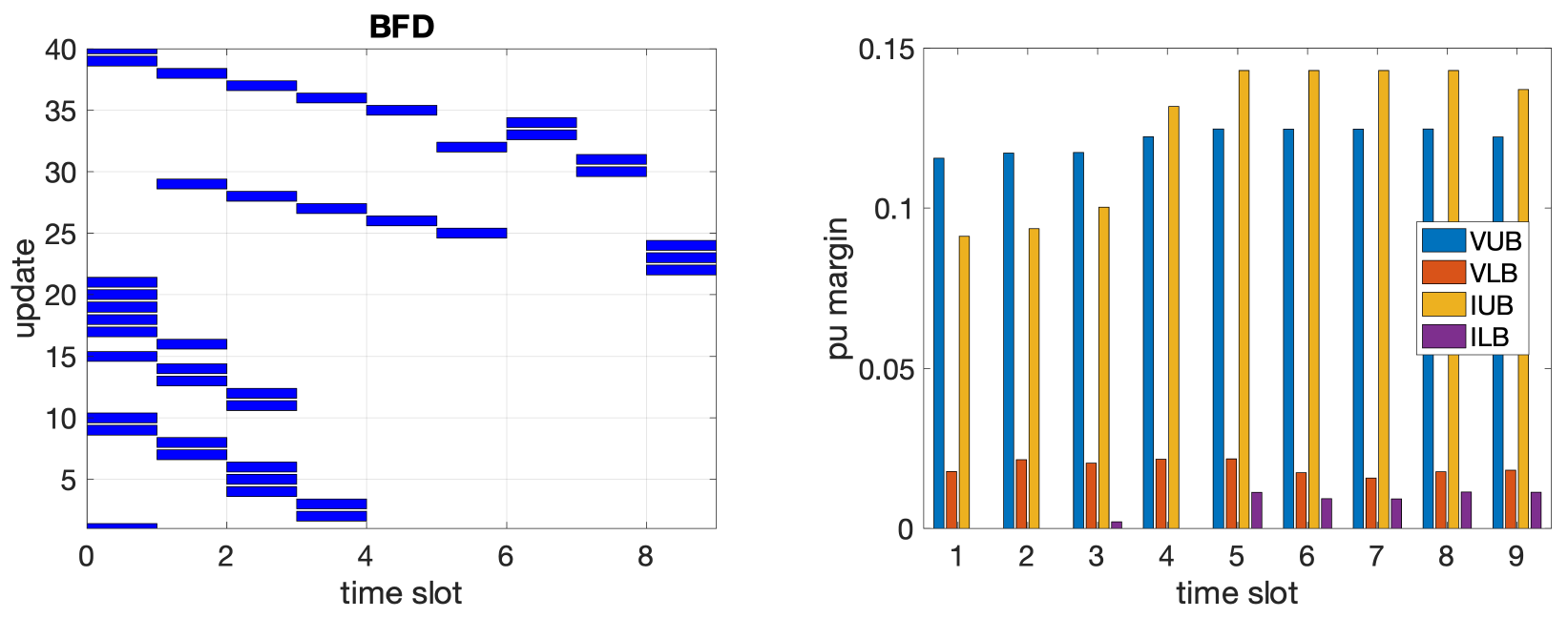}
    \caption{Results of CIGRE heavy load scenario case study concerning \eqref{opt:bin_packing}. Left: Gantt diagram. Right: voltage and current margins.}
    \label{fig:CIGRE_HL}
\end{figure}
Fig.~\ref{fig:CIGRE_HL} (right) also shows the voltage and current margins for each of the time slots defined as follows: for any time slot with assigned update $\mathcal{I} \subseteq \mathcal{N}$, the following problems are solved
\begin{subequations} \label{opt:vi_bounds_I}
    \begin{align}
        \nu^L_n(\mathcal{I}) &:= \min\limits_{p, q, P, Q, \nu, \ell} \; \nu_n \;\; \text{s.t.~\eqref{eqn:DF}, $(p,q) \in \mathcal{S}(\mathcal{I})$ in \eqref{eqn:SI}} \\
        \nu^U_n(\mathcal{I}) &:= \max\limits_{p, q, P, Q, \nu, \ell} \; \nu_n \;\; \text{s.t.~\eqref{eqn:DF}, $(p,q) \in \mathcal{S}(\mathcal{I})$ in \eqref{eqn:SI}} \\
        \ell^U_n(\mathcal{I}) &:= \max\limits_{p, q, P, Q, \nu, \ell} \; \ell_n \;\; \text{s.t.~\eqref{eqn:DF}, $(p,q) \in \mathcal{S}(\mathcal{I})$ in \eqref{eqn:SI}}
    \end{align}
\end{subequations}
to obtain the worst-case voltage for each non-reference bus $n \in \mathcal{N}$ and the worst-case current for each line $n \in \mathcal{L}$ for $\mathcal{I}$ associated with the time slot. Then, the worst-case quantities from \eqref{opt:vi_bounds_I} are compared with the safety limits in \eqref{eqn:lv_bounds} to obtain the relevant margins (nonnegative values mean limit satisfaction). It can be seen that the rollout schedule obtained by solving the proposed problem in \eqref{opt:bin_packing} with Algorithm~\ref{alg:BFD} indeed satisfies all safety limits.

As a comparison, the linearized DistFlow version of \eqref{opt:bin_packing} is considered. The following problem (i.e., an extension of the linearized DistFlow rollout problem in \cite{de2020minimum}) is solved to obtain the corresponding rollout schedule:
\begin{align} \label{opt:bin_packing_LDF}
    \begin{split}
        \text{minimize} & \quad T \\
        \text{subject to} & \quad \text{$\mathcal{I}_1, \mathcal{I}_2, \ldots, \mathcal{I}_T$ partition $\mathcal{N}$} \\
        & \quad {\bf \tilde{H}} {\bf 1}_{\mathcal{I}_j} \le {\bf \tilde{b}}, \quad j \in \{1, \ldots, T\}
    \end{split}
\end{align}
where
\begin{align*}
    {\bf \tilde{H}} = \begin{bmatrix}
        W^{\text{vub}} \\ \tilde{W}^{\text{vlb}}
    \end{bmatrix}, \quad {\bf \tilde{b}} = \begin{bmatrix}
        \overline{\nu} - \hat{\nu}^{\text{nom}} \\ \hat{\nu}^{\text{nom}} - \underline{\nu}
    \end{bmatrix}
\end{align*}
with $W^{\text{vub}}$ and $\hat{\nu}^{\text{nom}}$ defined in \eqref{eqn:vub_vW} and $\tilde{W}^{\text{vlb}} := 2 X D_C + 2 R D_{\hat{p}^G} + 2 X D_{\hat{q}^G}$ (i.e., the matrix in \eqref{eqn:vlb_W} with the last term removed). Fig.~\ref{fig:CIGRE_HL_LDF}
\begin{figure}[!tb]
    \centering
    \includegraphics[width=88mm]{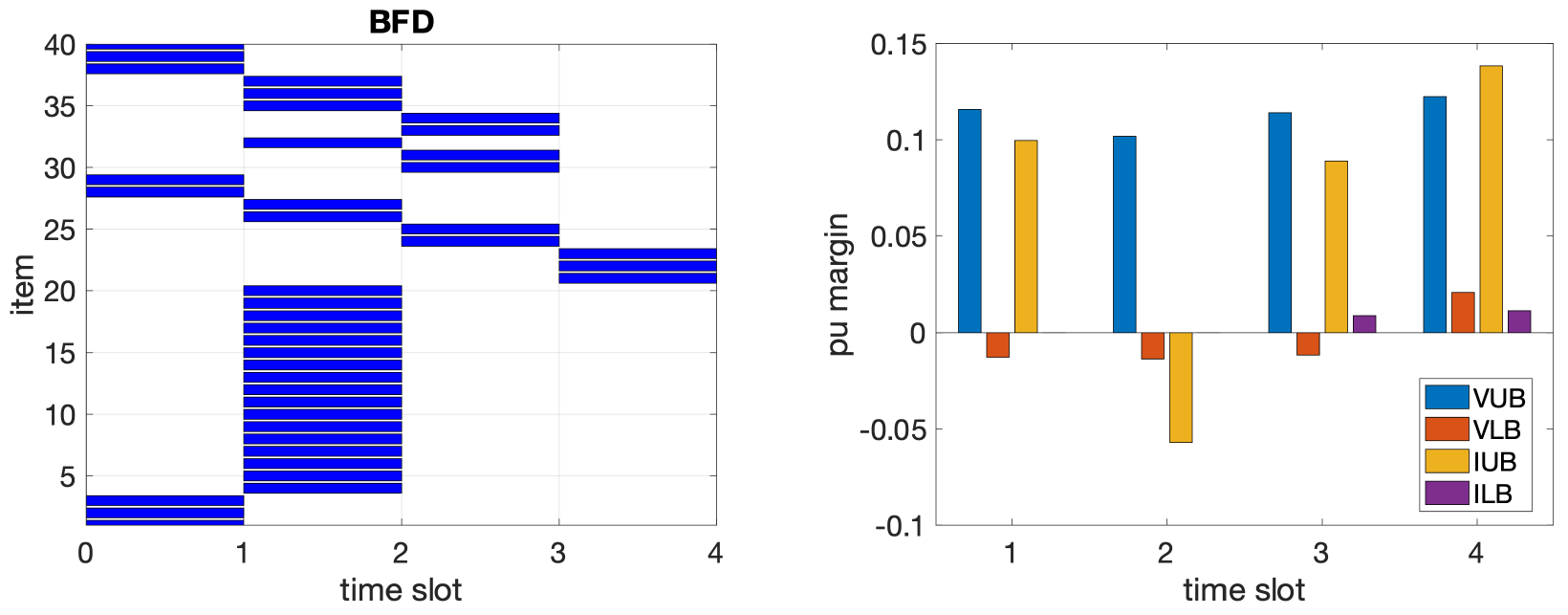}
    \caption{Results of CIGRE heavy load scenario case study concerning \eqref{opt:bin_packing_LDF}. Left: Gantt diagram. Right: voltage and current margins.}
    \label{fig:CIGRE_HL_LDF}
\end{figure}
shows the Gantt diagram and the margins of the rollout schedule obtained by solving \eqref{opt:bin_packing_LDF} using Algorithm~\ref{alg:BFD}. The linearized DistFlow schedule requires four versus nine time slots required by the proposed problem in \eqref{opt:bin_packing}. However, Fig.~\ref{fig:CIGRE_HL_LDF} shows negative voltage lower limit margins for the first three time slots, as well as negative current upper limit margin for the second time slot. This indicates voltage and current limit violations, as illustrated in Fig.~\ref{fig:CIGRE_HL_violation_LDF}.
\begin{figure}[!tb]
    \centering
    \includegraphics[width=88mm]{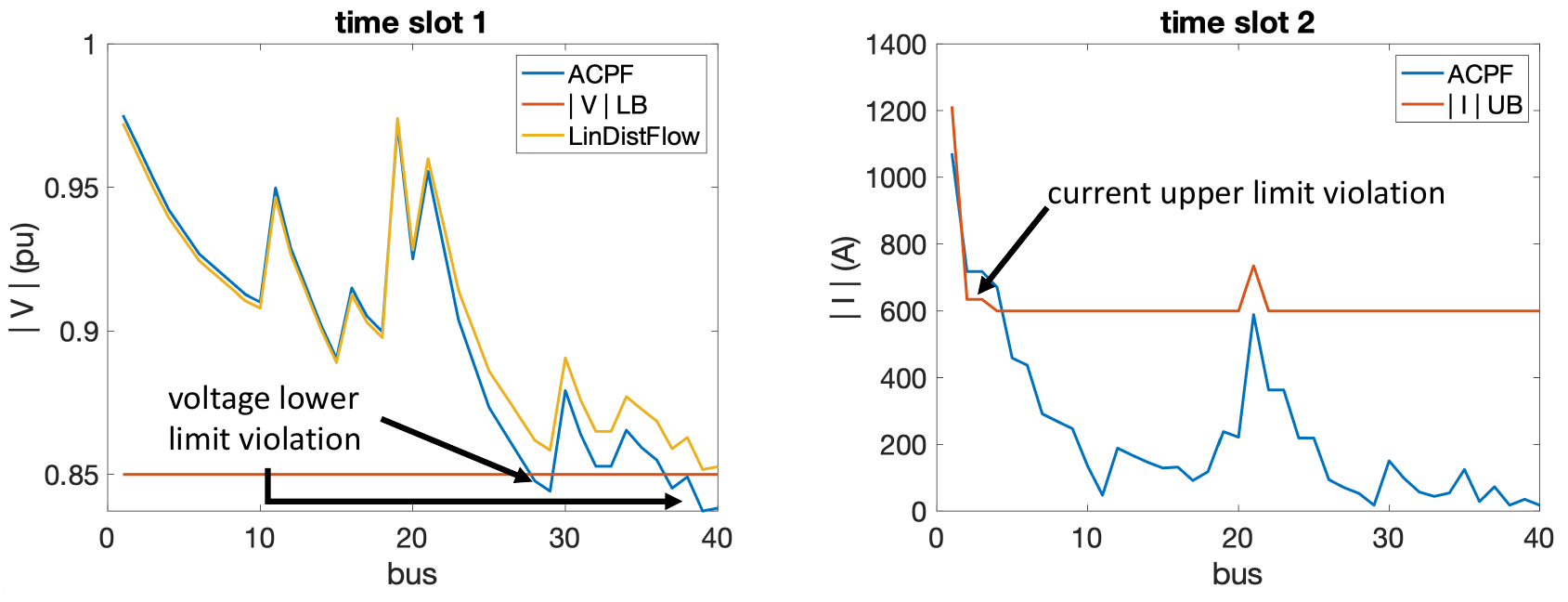}
    \caption{Voltage and current profiles showing safety limit violations of the rollout schedule due to \eqref{opt:bin_packing_LDF}. Left: voltage profile of time slot 1. Right: current profile of time slot 2. In both subfigures, ACPF denotes the voltage and current profiles obtained from AC power flow simulation using Matpower \cite{zimmerman2010matpower} with some power injections $(p,q) \in \mathcal{S}(\mathcal{I})$ for $\mathcal{I}$ conforming to the time slot. In the left, LinDistFlow denotes the minimum voltage due to linearized DistFlow equation (i.e., $(\hat{\nu}^{\text{nom}} + W^{\text{vub}} {\bf 1}_{\mathcal{I}})^{1/2}$), which suggests false limit satisfaction.}
    \label{fig:CIGRE_HL_violation_LDF}
\end{figure}
This justifies the apparently more complicated safety constraints proposed in \eqref{eqn:Hb} for the rollout scheduling problem, as opposed to the simpler linearized DistFlow version introduced earlier in \cite{de2020minimum}.

\subsection{CIGRE 44-bus LV Distribution Grid -- Significant DG}
Opposite to Section~\ref{subsec:CIGRE_HL}, this section considers the scenario with significant distributed generation and reverse power flows. In this scenario, the loads $\hat{p}^L$ and $\hat{q}^L$ retain their default values. The total inverter power rating is increased to 0.9 pu (versus total active power load of 0.69 pu). In addition, all inverters nominally generate active power according to their ratings (i.e., $\hat{p}^G = 0.9/N {\bf 1}$ pu and $\hat{q}^G = {\bf 0}$ pu). The voltage lower limits and current limits and the current upper limits are all uniform over their elements. They are respectively 0.9 pu, 1.2 pu and 800 A. Fig.~\ref{fig:CIGRE_HDG}
\begin{figure}[!tb]
    \centering
    \includegraphics[width=88mm]{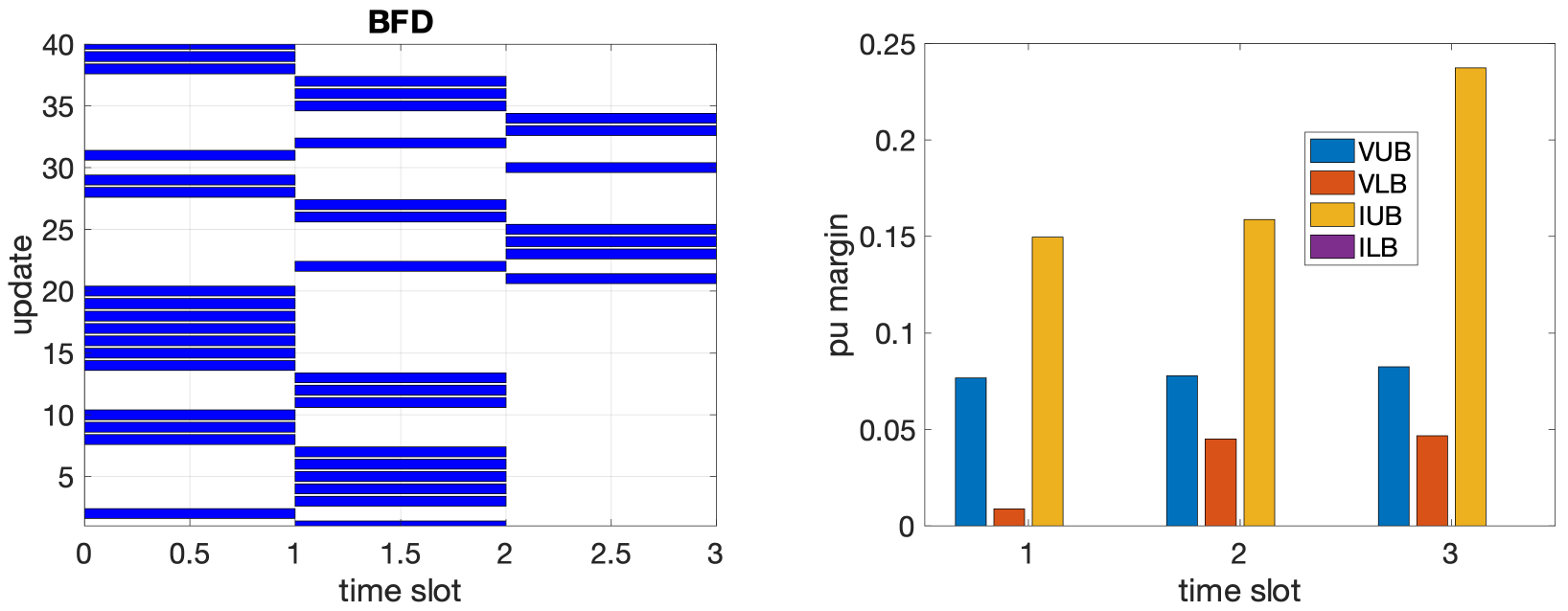}
    \caption{Results of CIGRE significant DG scenario case study concerning \eqref{opt:bin_packing}. Left: Gantt diagram. Right: voltage and current margins.}
    \label{fig:CIGRE_HDG}
\end{figure}
shows the Gantt diagram and the voltage and current margins for the rollout schedule obtained by the proposed method. In comparison, Fig.~\ref{fig:CIGRE_HDG_LDF}
\begin{figure}[!tb]
    \centering
    \includegraphics[width=88mm]{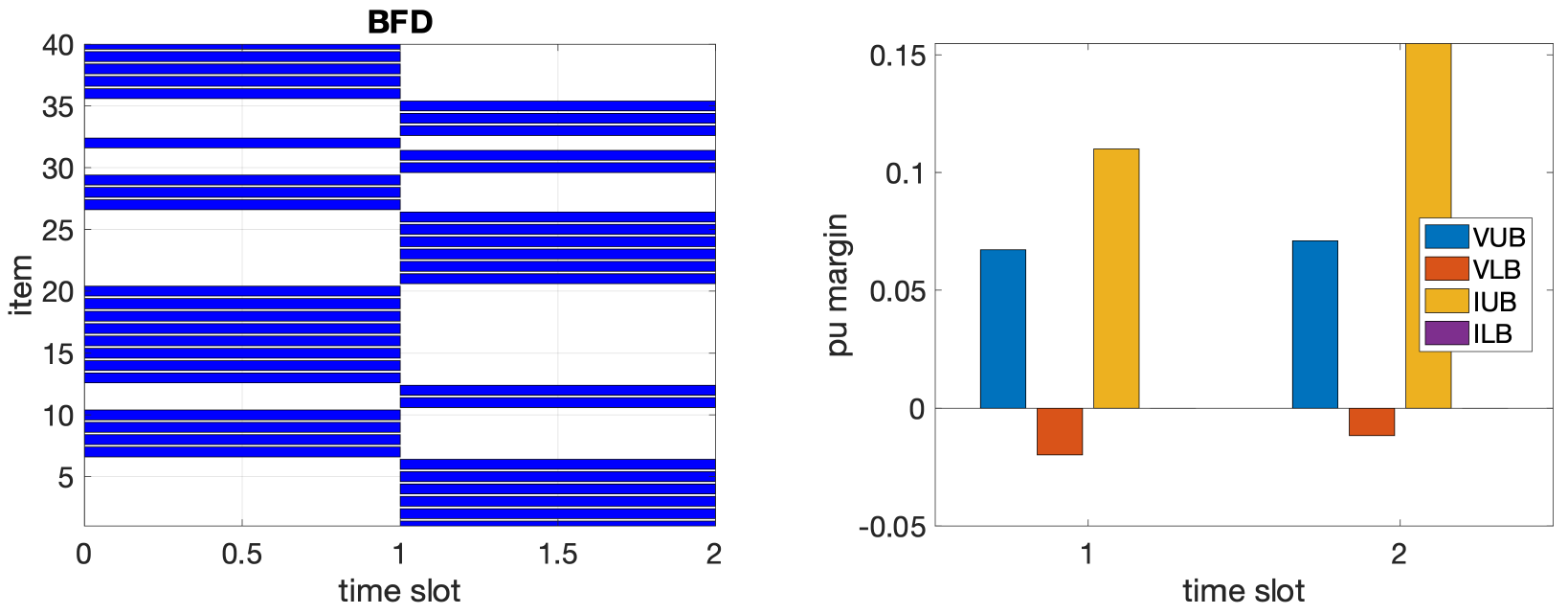}
    \caption{Results of CIGRE significant DG scenario case study concerning \eqref{opt:bin_packing_LDF}. Left: Gantt diagram. Right: voltage and current margins.}
    \label{fig:CIGRE_HDG_LDF}
\end{figure}
shows the analogous results for the rollout schedule obtained by solving the linearized DistFlow rollout scheduling problem in \eqref{opt:bin_packing_LDF}. The linearized version requires one fewer time slot. However, voltage lower limits are again violated.

\subsection{REDS 10,476-bus/83-feeder Distribution Grid} \label{subsec:10k}
This benchmark is the largest example from the repository in \cite{REDS_KA}, containing 10,476 non-reference buses and 83 feeders. We assume that inverters are installed in all non-reference buses with a uniform rating. The total rating is 68.4 pu, while the total active power load is 38.2 pu. The safety limits are, respectively, 0.9 pu, 1.1 pu and 600 A for voltage and current (uniform over all relevant elements). Solving the proposed rollout scheduling problem in \eqref{opt:bin_packing} with this benchmark yields a schedule with Gantt diagram and voltage and current margins shown in Fig.~\ref{fig:10k}.
\begin{figure}[!tb]
    \centering
    \includegraphics[width=88mm]{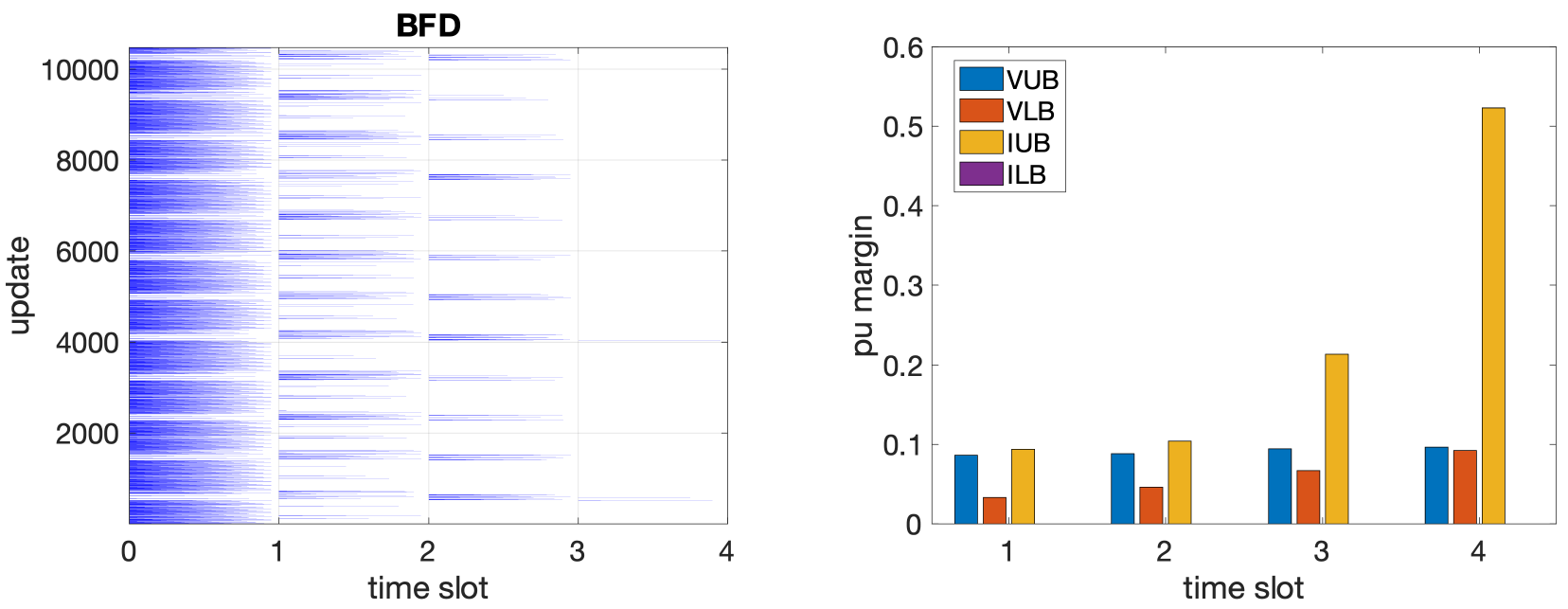}
    \caption{Results of 10,476-bus benchmark case study concerning \eqref{opt:bin_packing}. Left: Gantt diagram. Right: voltage and current margins.}
    \label{fig:10k}
\end{figure}
The schedule requires four time slots to arrange all 10,476 updates, and the total time to set up and solve the rollout problem is less than 3 seconds (with Algorithm~\ref{alg:BFD} requiring about 50\% of the time). It can be verified that it is impossible to assign all updates to one single time slot, excluding the trivial solution. To obtain a rollout schedule with four time slots using enumeration would require checking the second constraints in \eqref{opt:bin_packing} more than $5 \times 10^{14}$ times. Suppose each such examination requires $2 \times 10^{-5}$ seconds (based on the time recording in our experiment). The total enumeration would require more than 300 years, which is not suitable for real-time applications.

\subsection{Summary for all cases}
Instances of the rollout problem in \eqref{opt:bin_packing} derived from benchmarks from the REDS repository \cite{REDS_KA} and Matpower 7.1 \cite{zimmerman2010matpower} are solved using Algorithm~\ref{alg:BFD}. In all these cases, the voltage limits are 0.9 pu and 1.1 pu for all buses and the current upper limits are 600 A or above (depending on the nominal current due to default load setting). Typical load and inverter rating scenarios are assumed, which could be different from the scenario in Section~\ref{subsec:10k}. Table~\ref{tab:statistics} summarizes the results of the rollout schedules.
\begin{table}[!htb]
    \centering
    \caption{Computation time and number of time slots of rollout schedules obtained by the proposed method}
    \label{tab:statistics}
    \begin{tabular}{|c|c|c|c|c|c|}
    \hline
        case name & time & slots & case name & time & slots \\
        \hline
         REDS 13+3 & 4.98 ms & 1 & 33bw & 4.35 ms & 2 \\
         \hline
         REDS 29+1 & 4.77 ms & 2 & 51he & 6.08 ms & 1 \\
         \hline
         REDS 32+1 & 4.84 ms & 2 & 69 & 8.85 ms & 4\\
         \hline
         REDS 83+11 & 6.35 ms & 4 & 85 & 11.6 ms & 6\\
         \hline
         REDS 135+8 & 10.9 ms & 3 & 94pi & 20.5 ms & 11\\
         \hline
         REDS 201+3 & 30.7 ms & 4 & 118zh & 11.3 ms & 8\\
         \hline
         REDS 873+7 & 151 ms & 2 & 136ma & 10.2 ms & 3\\
         \hline
         REDS 10476+84 & 1.97 s & 2 & 141 & 20.1 ms & 5 \\
         \hline
    \end{tabular}
\end{table}
We note that the computation time for all cases, including the 10476-bus example, remains modest. In addition, the number of time slots used by the schedules are not excessive in comparison to the ``sequential schedules'' where each time slot holds exactly one update. This suggests, despite all approximations, that our proposed method is able to obtain safe and effective rollout schedules in real-time, even for distribution networks of substantial size.

\section{CONCLUSIONS}
Rapid system-level deployment and reconfiguration of control software is gaining prominence as power systems evolve to integrate large amounts of inverter-based distributed resources with little or zero inertia, thus requiring more dynamic control. As the power system is part of the critical infrastructure, its control software update should be carefully executed to prevent major reliability incidents. This requires proper accounting of the cyber-physical relationship between software behavior (e.g., update failure) and power system operational state (e.g., voltage and current). For radial distribution systems, the DistFlow equation accurately describes the underlying physics. However, the intrinsic nonlinearity renders the equation unfit for any real-time decision-making. The linearized DistFlow equation is generally a good balance between model complexity and fidelity. Nevertheless, as demonstrated in this paper, it may not be accurate enough to maintain system safety standard when conflicting operational goals are present (e.g., makespan minimization for software update rollout). This necessitates innovations in distribution system modeling exemplified by the proposed linearized relationships between software update failure and the worst-case voltage/current, which are theoretically justifiable with our detailed analysis and practically viable as demonstrated by benchmarks with complexity far exceeding the norm in the literature (e.g., 10476-bus vs 33-bus). Besides the software update rollout problem studied here, these linearized relationships can potentially enable online distribution system contingency monitoring and re-dispatch. In addition, to enable realistic large-scale applications computational innovations are needed. The reformulation of the rollout problem into a bin-packing problem amenable to efficient best-fit decreasing algorithm is one such example. In addition, the fixed point iteration procedure to quickly estimate the worst-case voltages and currents crucially enables the real-time large-scale application of the proposed rollout scheduling procedure. The investigation to apply the fixed point procedure to other situations (e.g., quick power flow analysis, contingency evaluation, etc.) appears to be worthwhile.





\section*{Appendix A -- Auxiliary Statements}
\begin{lemma} \label{thm:abr_max}
Let $a + j b \in \mathbb{C}$ and $r > 0$ be given. Also, let
\begin{align*}
    \theta := \begin{cases}
        \tan^{-1}(b/a), & \text{if $a \neq 0$ or $b \neq 0$} \\
        0, & \text{if $a = b = 0$}
    \end{cases}
\end{align*}
Then, the maximum objective value of the following problem
\begin{align} \label{opt:abr_max}
\begin{split}
    \underset{p,q}{\text{maximize}} & \quad |(a+jb)-(p+jq)| \\
    \text{subject to} & \quad p \ge 0, \;\;|p + jq| \le r
\end{split}
\end{align}
is 
\begin{align} \label{eqn:abr_max_value}
\max \big\{&|a+j(b+r)|, |a+j(b-r)|, |a+jb-r e^{j \theta}| \big\}
\end{align}
\end{lemma}
\begin{proof}
First we consider the special case where $a = b = 0$. In this case, the maximum of \eqref{opt:abr_max} is $r$ and all three terms in \eqref{eqn:abr_max_value} are equal to $r$. Thus, the statement is trivially true when $a = b = 0$. Next, we consider the case where at least one of $a$ or $b$ is not zero. Let $(p^*, q^*)$ be the argument of maximum of \eqref{opt:abr_max}. We claim that $|p^*+jq^*| = r$. Suppose not, then there exists $\Delta q \neq 0$ such that $|p^*+j(q^*+\Delta q)| = r$ and $|b-q^*-\Delta q| > |b-q^*|$. This implies that $(p^*, q^*+\Delta q)$ is a better feasible solution than $(p^*, q^*)$, leading to a contradiction. Hence, \eqref{opt:abr_max} can be re-parameterized as
\begin{align} \label{opt:abr_max1}
    \underset{\alpha \in [-\pi/2, \; \pi/2]}{\text{maximize}} \quad |r \cos(\alpha) - a + j (r \sin(\alpha) - b)|
\end{align}
It can be verified that the only stationary point of the objective for $-\pi/2 \le \alpha \le \pi/2$ is at $\tan^{-1}(b/a)$ (well-defined because at least one of $a$ or $b$ is nonzero). Therefore, the maximum of \eqref{opt:abr_max1} can be attained only at one of the three possibilities: $-\pi/2$, $\pi/2$ and $\tan^{-1}(b/a)$, which is $\theta$ according to the statement definition. Thus, the maximum of \eqref{opt:abr_max1} must be one of the following: $|a+j(b+r)|$, $|a+j(b-r)|$ or $|a+jb-r(\cos(\theta)+j \sin(\theta))|$, same as \eqref{eqn:abr_max_value}.
\end{proof}

\begin{lemma} \label{thm:Ds_max}
    Let $\mathcal{D}$ be the descendant matrix defined in Section~\ref{subsec:distribution_system_notations}. Let $\hat{p}^L, \hat{q}^L, \hat{p}^G, \hat{q}^G \in \R^N$ be given. Then, for any $n \in \mathcal{N}$ and $\mathcal{I} \subseteq \mathcal{N}$, the optimal objective value of
    \begin{align} \label{opt:Ds_max}
    \begin{split}
        \underset{p^G, q^G}{\text{maximize}} \quad & e_n^\top \big| \mathcal{D} (p + j q ) \big| \\
        \text{subject to} \quad & (p,q) \in \mathcal{S}(\mathcal{I}) \; \text{in \eqref{eqn:SI}}
    \end{split}
    \end{align}
    is
    \begin{align} \label{eqn:Ds_max_value}
        \max \Big\{&\!|a \! + \! j (b\!+\!r)|, |a \!+\! j (b\!-\!r)|, |a\!+\!j b \!-\! r e^{j \theta}|\! \Big\}
    \end{align}
    where
    \begin{align} \label{eqn:DS_max_parameters}
    \begin{split}
        &a := e_n^\top \mathcal{D} (\hat{p}^L - ({\bf 1} - {\bf 1}_{\mathcal{I}}) \odot \hat{p}^G ) \\
        &b := e_n^\top \mathcal{D} (\hat{q}^L - ({\bf 1} - {\bf 1}_{\mathcal{I}}) \odot \hat{q}^G ) \\
        &r := e_n^\top \mathcal{D} ({\bf 1}_{\mathcal{I}} \odot C) \\
        &\theta := \begin{cases} \tan^{-1}(b/a), & \text{if $a \neq 0$ or $b \neq 0$} \\ 0, & \text{if $a = b = 0$} \end{cases}
    \end{split}
    \end{align}
\end{lemma}
\begin{proof}
    Note that for any $x \in \R^N$, $e_n^\top \mathcal{D} x = \sum\limits_{m \in d(n)} x_m$. Hence, with $(p,q) \in \mathcal{S}(\mathcal{I})$ the objective function can be written as
    \begin{align*}
        & \left| \sum\limits_{m \in d(n)} (p^G_m - \hat{p}^L_m) + j \sum\limits_{m \in d(n)} (q^G_m - \hat{q}^L_m) \right| \\
        = & \left| \big( \sum\limits_{m \in d(n)} \hat{p}^L_m - \!\!\! \sum\limits_{m \in d(n) \setminus \mathcal{I}} \hat{p}^G_m \big) \! + \! j \big( \sum\limits_{m \in d(n)} \hat{q}^L_m - \!\!\! \sum\limits_{m \in d(n) \setminus \mathcal{I}} \hat{q}^G_m \big) \right. \\
        & \left. - \; \sum\limits_{m \in d(n) \cap \mathcal{I}} (p^G_m + j q^G_m) \right|
    \end{align*}
    According to the definition of $\mathcal{S}(\mathcal{I})$, for each $m \in \mathcal{I}$, the set of all possible $p^G_m + j q^G_m$ is the half-disk $\Omega(C_m) := \{(x,y) \mid x \ge 0, |x + j y| \le C_m\}$. Thus, the set of all possible sum $\sum\limits_{m \in d(n) \cap \mathcal{I}} (p^G_m + j q^G_m)$ is the Minkowski sum $\sum\limits_{m \in d(n) \cap \mathcal{I}} \Omega(C_m)$, which can be verified to be the half-disk $\Omega(\sum\limits_{m \in d(n) \cap \mathcal{I}} C_m)$. By noting that
    \begin{align*}
        & \sum\limits_{m \in d(n)} \hat{p}^L_m - \!\!\! \sum\limits_{m \in d(n) \setminus \mathcal{I}} \hat{p}^G_m = e_n^\top \mathcal{D} (\hat{p}^L - ({\bf 1} - {\bf 1}_{\mathcal{I}}) \odot \hat{p}^G ) = a \\
        & \sum\limits_{m \in d(n)} \hat{q}^L_m - \!\!\! \sum\limits_{m \in d(n) \setminus \mathcal{I}} \hat{q}^G_m = e_n^\top \mathcal{D} (\hat{q}^L - ({\bf 1} - {\bf 1}_{\mathcal{I}}) \odot \hat{q}^G ) = b \\
        & \sum\limits_{m \in d(n) \cap \mathcal{I}} C_m = e_n^\top \mathcal{D} ({\bf 1}_{\mathcal{I}} \odot C) = r
    \end{align*}
    problem~\eqref{opt:Ds_max} can be rewritten as maximizing $|a + j b - (x + j y)|$ with respect to $x$ and $y$, such that $x \ge 0$ and $|x + j y| \le r$. According to Lemma~\ref{thm:abr_max}, the maximum objective value is given in \eqref{eqn:Ds_max_value}.
\end{proof}

\section*{Appendix B -- Proof of Proposition~\ref{thm:fixed_point}}
It suffices to show the following induction: suppose for some $k$ it holds that $v^k \le (\nu^L)^{1/2}$ and $i^k \ge (\ell^U)^{1/2}$, then $v^{k+1} \le (\nu^L)^{1/2}$ and $i^{k+1} \ge (\ell^U)^{1/2}$.

For any $n \in \mathcal{N}$, according to \eqref{opt:vi_bounds} and \eqref{eqn:DF_vi}, there exist some optimizing pair $(p^{G(n)}, q^{G(n)}) \in \R^N \times \R^N$ satisfying \eqref{eqn:inverter_constraints} with $(p^G, q^G)$ interpreted as $(p^{G(n)}, q^{G(n)})$ such that $\nu^L_n = v_n^2$ with $v$ (together with $i$) satisfying
\begin{align*}
    \begin{split}
        v &= \big( \nu_0 {\bf 1} + 2 R (p^{G(n)} - \hat{p}^L) + 2 X (q^{G(n)} - \hat{q}^L) + M i^2 \big)^{1/2} \\
        i &= v^{-1} \odot\big| \mathcal{D}((p^{G(n)} - \hat{p}^L) \! + \! j (q^{G(n)} - \hat{q}^L)) \! - \! (\mathcal{D} \! - \! I) D_z i^2 \big|
    \end{split}
\end{align*}
We note that all entries of $M$ are non-positive. This, together with $i^k \ge (\ell^U)^{1/2}$ by the induction assumption and $(\ell^U)^{1/2} \ge i$ by \eqref{opt:vi_bounds}, implies that $M i^2 \ge M (i^k)^2$. Thus, for any $n \in \mathcal{N}$
\begin{align*}
    &(\nu^L_n)^{1/2} \\
    = & \big( \nu_0 + e_n^\top \big( 2 R (p^{G(n)} - \hat{p}^L) + 2 X (q^{G(n)} - \hat{q}^L) + M i^2\big) \big)^{1/2} \\
    \ge & \big( \nu_0 \! + \! e_n^\top \big( 2 R (p^{G(n)} - \hat{p}^L) \! + \! 2 X (q^{G(n)} - \hat{q}^L) \! + \! M (i^k)^2\big) \big)^{1/2} \\
    \ge & \min_{\text{$(p^G, q^G)$ satisfying \eqref{eqn:inverter_constraints}}} \big\{ \big( \nu_0 + e_n^\top \big( 2 R (p^G - \hat{p}^L) \\
    &+ 2 X (q^G - \hat{q}^L) + M (i^k)^2 \big) \big)^{1/2} \big\} \\
    = & e_n^\top \big(\nu_0 {\bf 1} \! - \! 2 R \hat{p}^L - 2 X (\hat{q}^L \! + \! C) \! + \! M (i^k)^2 \big)^{1/2} \\
    = & v_n^{k+1}
\end{align*}

Similarly, there exists some $(p^{g(n)}, q^{g(n)}) \in \R^N \times \R^N$ satisfying \eqref{eqn:inverter_constraints} with $(p^G, q^G)$ interpreted as $(p^{g(n)}, q^{g(n)})$ such that $\ell^U_n = i_n^2$ with $i$ (together with $v$) satisfying
\begin{align*}
    \begin{split}
        v &= \big( \nu_0 {\bf 1} + 2 R (p^{g(n)} - \hat{p}^L) + 2 X (q^{g(n)} - \hat{q}^L) + M i^2 \big)^{1/2} \\
        i &= v^{-1} \odot\big| \mathcal{D}((p^{g(n)} - \hat{p}^L) \! + \! j (q^{g(n)} - \hat{q}^L)) \! - \! (\mathcal{D} \! - \! I) D_z i^2 \big|
    \end{split}
\end{align*}
Since all entries of $(\mathcal{D}-I)D_r$ and $(\mathcal{D}-I)D_x$ are nonnegative, $i^k \ge (\ell^U)^{1/2}$ and $v^k \le (\nu^L)^{1/2}$ by the induction assumption, $(\ell^U)^{1/2} \ge i$ and $(\nu^L)^{1/2} \le v$ by \eqref{opt:vi_bounds}, it holds that (a) $v \ge v^k$, (b) $(\mathcal{D}-I)D_r i^2 \le (\mathcal{D}-I)D_r (i^k)^2$ and (c) $(\mathcal{D}-I)D_x i^2 \le (\mathcal{D}-I)D_x (i^k)^2$. Hence, for any $n \in \mathcal{N}$
\begin{align*}
    &(\ell_n^U)^{1/2} \\
    = & v_n^{-1} e_n^\top \big| \mathcal{D}((p^{g(n)} - \hat{p}^L) \! + \! j (q^{g(n)} - \hat{q}^L)) \! - \! (\mathcal{D} \! - \! I) D_z i^2 \big| \\
    \le & {(v_n^k)}^{-1} \!\! e_n^\top \big(\big| \mathcal{D}((p^{g(n)} \!\! - \! \hat{p}^L) \!\! + \!\! j (q^{g(n)} \!\! - \! \hat{q}^L)) \big| \!\! + \!\! \big|(\mathcal{D} \!\! - \!\! I) D_z (i^k)^2 \big| \big) \\
    \le & \max_{\text{$(p^g, q^g)$ satisfying \eqref{eqn:inverter_constraints}}} {(v_n^k)}^{-1} \big| e_n^\top \mathcal{D} ((p^g - \hat{p}^L) + j (q^g - \hat{q}^L) ) \big| \\
    & + {(v_n^k)}^{-1} e_n^\top \big|(\mathcal{D} - I) D_z (i^k)^2 \big|
\end{align*}
According to Lemma~\ref{thm:Ds_max} with $\mathcal{I} = \mathcal{N}$, the first term above can be rewritten as $\bar{S}_n := \max\{|a_n + j (b_n + r_n)|, |a_n + j (b_n - r_n)|, |a_n + j b_n - r_n (\cos(\theta_n) + j \sin(\theta_n)|\}$, where $a_n = e_n^\top \mathcal{D} \hat{p}^L$, $b_n = e_n^\top \mathcal{D} \hat{q}^L$, $r_n = e_n^\top \mathcal{D}C$ and $\theta_n = \tan^{-1}(b_n/a_n)$ if at least one of $a_n$ and $b_n$ is nonzero (otherwise $\theta_n = 0$). Therefore,
\begin{align*}
    (\ell_n^U)^{1/2} \le {(v_n^k)}^{-1} \left( \bar{S}_n + e_n^\top \big|(\mathcal{D} - I) D_z (i^k)^2 \big| \right) = i_n^{k+1} \;\;\; \qed
\end{align*}


\bibliographystyle{IEEEtran}
\bibliography{references}

\begin{IEEEbiographynophoto}
{Kin Cheong Sou} received a Ph.D. degree in Electrical Engineering and Computer Science at Massachusetts Institute of Technology in 2008. From 2008 to 2010 he was a postdoctoral researcher at Lund University, Lund, Sweden. From 2010 to 2012 he was a postdoctoral researcher at KTH Royal Institute of Technology, Stockholm, Sweden. Between 2013 and 2016 he was an assistant professor with the department of Mathematical Sciences, Chalmers University of Technology and the University of Gothenburg, Sweden. Dr. Sou is now an associate professor with the department of Electrical Engineering at the National Sun Yat-sen University in Taiwan. His research interests include optimization and system analysis with power systems applications.
\end{IEEEbiographynophoto}
\vspace{-10mm}
\begin{IEEEbiographynophoto}
{Henrik Sandberg} is Professor at the Division of Decision and Control Systems, KTH Royal Institute of Technology, Stockholm, Sweden. He received the M.Sc. degree in engineering physics and the Ph.D. degree in automatic control from Lund University, Lund, Sweden, in 1999 and 2004, respectively. From 2005 to 2007, he was a Post-Doctoral Scholar at the California Institute of Technology, Pasadena, USA. In 2013, he was a Visiting Scholar at the Laboratory for Information and Decision Systems (LIDS) at MIT, Cambridge, USA. He has also held visiting appointments at the Australian National University and the University of Melbourne, Australia. His current research interests include security of cyber-physical systems, power systems, model reduction, and fundamental limitations in control. Dr. Sandberg was a recipient of the Best Student Paper Award from the IEEE Conference on Decision and Control in 2004, an Ingvar Carlsson Award from the Swedish Foundation for Strategic Research in 2007, and a Consolidator Grant from the Swedish Research Council in 2016. He has served on the editorial boards of IEEE Transactions on Automatic Control and the IFAC Journal Automatica. He is Fellow of the IEEE.
\end{IEEEbiographynophoto}

\end{document}